\documentclass[lettersize,journal]{IEEEtran}
\pdfoutput=1
\usepackage{amsmath,amssymb,amsfonts}
\usepackage{array}
\usepackage[caption=false,font=normalsize,labelfont=sf,textfont=sf]{subfig}
\usepackage{textcomp}
\usepackage{stfloats}
\usepackage{url}
\usepackage{verbatim}
\usepackage{graphicx}
\usepackage{subcaption}
\usepackage{subfig}  % Use subfig instead of subcaption
\usepackage{graphicx}
\usepackage{titlesec}
\usepackage{caption}
\usepackage{subcaption}
\usepackage{cite}
\usepackage{booktabs}
\usepackage{xcolor}
\usepackage{epsfig}
\usepackage{epstopdf}
\usepackage{acronym}
\graphicspath{{fig/}}
\usepackage{mathtools}
\usepackage{amsthm}
\usepackage{bm}
\usepackage{comment}
 \usepackage{subcaption}
 \usepackage{setspace}
  \usepackage{graphicx}
  \usepackage{lipsum}
  \usepackage[absolute,overlay]{textpos}
  \usepackage{upgreek}
  \usepackage{siunitx}
  \usepackage[greek,english]{babel}

\usepackage{ragged2e}
\usepackage{textgreek}

\newacro{ACDD}{Alamouti with cyclic delay diversity}
\newacro{URLLC}{ultra-reliable low-latency communications}
\newacro{3GPP}{third generation partnership project}
\newacro{PHY}{physical layer}
\newacro{MIMO}{multiple-input multiple-output}
\newacro{SIMO}{single-input multiple-output}
\newacro{MISO}{multiple-input single-output}
\newacro{SISO}{single-input single-output}
\newacro{MRC}{maximum-ratio combining}
\newacro{SNR}{signal-to-noise ratio}
\newacro{CP}{cyclic prefix}
\newacro{CDD}{cyclic delay diversity}
\newacro{FSC}{frequency-selective channel}
\newacro{STC}{space-time coding}
\newacro{FFT}{fast Fourier transform}
\newacro{LMMSE}{linear minimum mean-squared error}
\newacro{FER}{frame error rate}
\newacro{OFDM}{orthogonal frequency division multiplexing}
\newacro{OCDM}{orthogonal chirp division multiplexing}
\newacro{FSC}{frequency-selective channel}
\newacro{CSI}{channel state information}
\newacro{LMMSE-PIC}{linear minimum mean squared error with parallel interference cancellation}
\newacro{PFE}{perfect-feedback equalizer}
\newacro{FD}{full-duplex}
\newacro{PDP}{power delay profile}
\newacro{PDF}{probability density function}
\newacro{DFT}{discrete Fourier transform}
\newacro{SDFT}{sparse DFT}
\newacro{ICI}{inter-carrier interference}
\newacro{OTFS}{orthogonal time frequency space}
\newacro{AWGN}{additive white Gaussian noise}
\newacro{SWH}{sparse Walsh-Hadamard}
\newacro{LLR}{log-likelihood ratio}
\newacro{PMF}{probability mass function}
\newacro{CRC}{cyclic redundancy check}
\newacro{PAM}{pulse amplitude modulation}
\newacro{QAM}{quadrature amplitude modulation}
\newacro{FWHT}{fast Walsh-Hadamard transform}
\newacro{MAP}{maximum a-posteriori}
\newacro{SC}{single-carrier}
\newacro{ISI}{inter-symbol interference}
\newacro{ZP}{zero-padding}
\newacro{EVD}{eigenvalue decomposition}
\newacro{BCJR}{Bahl, Cocke, Jelinek, and Raviv}
\newacro{WHT}{Walsh-Hadamard transform}
\newacro{APP}{a-posteriori probability}
\newacro{SILE-EPIC}{self-iterated linear equalizer with expectation propagation}
\newacro{EP}{expectation propagation}
\newacro{i.i.d.}{independent and identically distributed}
\newacro{CWCU}{component wise conditionally unbiased}
\newacro{MSE}{mean squared error}
\newacro{EXIT}{extrinsic information transfer}
\newacro{MI}{mutual information}
\newacro{PAPR}{peak-to-average power ratio}
\newacro{DFT-s}{discrete Fourier transform-spread}
\newacro{AMP}{approximate message passing}
\newacro{GAMP}{generalized \ac{AMP}}
\newacro{VAMP}{vector \ac{AMP}}
\newacro{RSC}{recursive systematic convolutional}
\newacro{QPSK}{quadrature phase-shift keying}
\newacro{CFAR}{constant false alarm rate}
\newacro{PD}{probability of detection}
\newacro{PFA}{probability of false alarm}
\newacro{RV}{random variable}
\newacro{CDF}{cumulative distribution function}
\newacro{HD-ZP}{half-duplex ZP}
\newacro{FD-CP}{full-duplex ZP}
\newacro{DFRC}{dual-function radar communication}
\newacro{SINR}{signal-to-interference noise ratio}
\newacro{ISAC}{integrated sensing and communication}
\newacro{SI}{self-interference}
\newacro{RSI}{residual self-interference}
\newacro{ADC}{analog-to-digital converter}
\newacro{DAC}{digital-to-analog converter}
\newacro{ED}{energy-detection}
\newacro{IDFT}{inverse discrete Fourier Transform}
\newacro{SFFT}{symplectic finite Fourier transform }
\newacro{CRB}{Cram{\'{e}}r-Rao bound}
\newacro{ZC}{Zadoff-Chu}
\newacro{RMSE}{root mean square error}
\newacro{UW}{unique word}
\newacro{GFDM}{generalized frequency division multiplexing}
\newacro{RRC}{root-raised cosine}
\newacro{UB}{upper bound}
\newacro{CEF}{channel estimation field}
\newacro{TRX}{transceiver}
\newacro{IF}{intermediate frequency}
\newacro{RF}{radio frequency}
\newacro{FPGA}{field programmable gate arrays}
\newacro{SDR}{software-defined radio}
\newacro{UWB}{ultra wideband}
\newacro{PCB}{printed circuit board}
\newacro{SMA}{SubMiniature version A}
\newacro{MUSIC}{multiple signal classification}
\newacro{CIR}{channel impulse response}
\newacro{FR}{Frequency Range}
\newacro{mmWave}{millimeter wave}
\newacro{LoS}{line-of-sight}
\newacro{ULA}{uniform linear array}
\newacro{AoA}{angle-of-arrival}
\newacro{AoD}{angle-of-departure}
\newacro{FIM}{Fisher Information Matrix}
\newacro{TDoA}{time difference of arrival}
\newacro{LM}{Levenberg-Marquardt}
\newacro{RSSI}{received signal strength indicator}
\newacro{OMP}{orthogonal matching pursuit}
\newacro{NLoS}{non line-of-sight}
\newacro{FDoA}{frequency difference-of-arrival}
\newacro{TSNR}{threshold SNR}
\newacro{ML}{maximum-likelihood}
\newacro{GPS}{global positioning system}

\def\BibTeX{{\rm B\kern-.05em{\sc i\kern-.025em b}\kern-.08em
		T\kern-.1667em\lower.7ex\hbox{E}\kern-.125emX}}
	
\usepackage{soul,color}

\usepackage{tikz}
\usepackage{pgfplots}
\usetikzlibrary{shapes,arrows}
\usetikzlibrary{positioning,calc}
\usetikzlibrary{decorations.pathreplacing,calligraphy}
\usetikzlibrary{arrows.meta}
\usepackage{siunitx}

\tikzset{add/.style n args={4}{
		minimum width=3mm,
		path picture={
			\draw[black] 
			(path picture bounding box.south east) -- (path picture bounding box.north west)
			(path picture bounding box.south west) -- (path picture bounding box.north east);
			\node at ($(path picture bounding box.south)+(0,0.13)$)     {\tiny #1};
			\node at ($(path picture bounding box.west)+(0.13,0)$)      {\tiny #2};
			\node at ($(path picture bounding box.north)+(0,-0.13)$)        {\tiny #3};
			\node at ($(path picture bounding box.east)+(-0.13,0)$)     {\tiny #4};
		}
	}
}

\tikzset{add2/.style n args={4}{
		minimum width=1mm,
		path picture={
			\draw[black] 
			(path picture bounding box.south) -- (path picture bounding box.north)
			(path picture bounding box.west) -- (path picture bounding box.east);
			\node at ($(path picture bounding box.south)+(0,0.13)$)     {\tiny #1};
			\node at ($(path picture bounding box.west)+(0.13,0)$)      {\tiny #2};
			\node at ($(path picture bounding box.north)+(0,-0.13)$)        {\tiny #3};
			\node at ($(path picture bounding box.east)+(-0.13,0)$)     {\tiny #4};
		}
	}
}

\usepackage{pgfplots}
\usepgfplotslibrary{polar}
\usepackage{tikz}

\newtheorem{proposition}{Proposition}

\counterwithin{corollary}{proposition} % makes Corollary 1.1, 1.2, etc.

\definecolor{applegreen}{rgb}{0.55, 0.71, 0.0}
\definecolor{awesome}{rgb}{1.0, 0.13, 0.32}
\definecolor{azure(colorwheel)}{rgb}{0.0, 0.5, 1.0}
\definecolor{darklavender}{rgb}{0.45, 0.31, 0.59}
\definecolor{cyan(process)}{rgb}{0.0, 0.72, 0.92}
\definecolor{brightmaroon}{rgb}{0.76, 0.13, 0.28}
\definecolor{ao(english)}{rgb}{0.0, 0.5, 0.0}
\definecolor{brightturquoise}{rgb}{0.03, 0.91, 0.87}
\definecolor{bondiblue}{rgb}{0.0, 0.58, 0.71}
\definecolor{atomictangerine}{rgb}{1.0, 0.6, 0.4}
\definecolor{classicrose}{rgb}{0.98, 0.8, 0.91}
\definecolor{copperrose}{rgb}{0.6, 0.4, 0.4}

\usepackage{etoolbox}

\usepackage{algorithm}
\usepackage{algpseudocode}

\algnewcommand{\Input}[1]{\State \textbf{Input:} #1}
\algnewcommand{\Output}[1]{\State \textbf{Output:} #1}
\AtBeginDocument{\everymath{\small} \everydisplay{\small}}
\newcommand{\sizefont}{\small}

\makeatletter
\patchcmd{\@maketitle}
{\centering}
{\centering
	\begin{minipage}{\textwidth}
		R. Bomfin, M. Mezzavilla, S. Rangan and M. Chafii, ``Near Field Multi-Band Localization: CRB, Efficient Estimator, and Threshold SNR,'' Submitted to \textit{IEEE Transactions on Wireless Communications}, March 2026.
	\end{minipage}\par\vspace{0.8em}
}
{}{}
\makeatother

\begin{document}
\title{Near Field Multi-Band Localization: CRB, Efficient Estimator, and Threshold SNR}
%(NYU) Abu Dhabi, 129188, UAE (email: roberto.bomfin@nyu.edu).}
%\thanks{Marwa Chafii is with Engineering Division, New York University (NYU)
%Abu Dhabi, 129188, UAE and NYU WIRELESS, NYU Tandon School of
%Engineering, Brooklyn, 11201, NY, USA (email: marwa.chafii@nyu.edu).}}
\author{Roberto Bomfin, Marco Mezzavilla, Sundeep Rangan, and  Marwa Chafii
\thanks{Roberto Bomfin is with the Engineering Division, New York University (NYU) Abu Dhabi, 129188, UAE (email: roberto.bomfin@nyu.edu).}
\thanks{Marco Mezzavilla is with the Dipartimento di Elettronica, Informazione e Bioingegneria (DEIB), Politecnico di Milano, Milan, Italy (email: marco.mezzavilla@polimi.it).}
\thanks{Sundeep Rangan is with NYU WIRELESS, NYU Tandon School of Engineering, Brooklyn, 11201, NY, USA (email: ar7655@nyu.edu, srangan@nyu.edu).}
\thanks{Marwa Chafii is with the Engineering Division, New York University (NYU) Abu Dhabi, 129188, UAE, and with NYU WIRELESS, NYU Tandon School of
Engineering, Brooklyn, 11201, NY, USA (email: marwa.chafii@nyu.edu).}}
\vspace{-0.5cm}
\maketitle

\begin{abstract}
This paper presents a theoretical framework for multi-band localization for a single-path single-input multiple-output (SIMO) system. 
We derive closed–form Cramér–Rao bounds (CRBs) for angle-of-arrival (AoA) and distance for uniform linear arrays (ULAs), and an intermediate matrix-form formulation for arbitrary array shapes. 
We also develop benchmark single- and multi-band maximum-likelihood (ML) estimators for AoA-Distance, leveraging a structured Levenberg-Marquardt (LM) refinement procedure.
A key contribution is an analytical characterization of the threshold SNR (TSNR) for the proposed estimators.
This is the SNR threshold at which the estimator transitions from ``off the chart'' to CRB-approaching performance, for both TDoA and distance estimation. 
Numerical simulations confirm that the proposed single- and multi-band estimators achieve the CRB at SNRs above the predicted TSNR, and that multi-band processing simultaneously improves estimation accuracy and reduces SNR requirements. The resulting framework provides a rigorous foundation for next-generation multi-band localization and can be readily extended to elevation estimation, distributed arrays, and multi-path environments.

\end{abstract}

\section{Introduction}
Wireless localization is expected to become a defining capability of next-generation 6G networks \cite{Stylianos,Henk5GPart1}.
High-precision positioning is essential for emerging applications such as autonomous robots, extended-reality (XR), industrial automation, and integrated sensing and communication (ISAC) systems \cite{chafii2023twelve,LiuJSAC}.
With modern systems moving toward higher frequencies, larger antenna arrays, new spectrum, and joint sensing–communication operation \cite{kang2024cellular}, both industry and academia are investigating how these new resources can be exploited to achieve centimeter-level localization.

Multi-band processing across wide and heterogeneous spectral resources is another promising direction.
The recently proposed Frequency Range~3 (FR3), spanning approximately 7 to 24,GHz, offers an attractive compromise between FR1 and FR2 with improved coverage, penetration, and spatial resolution \cite{kang2024cellular,mezzavilla2024frequency}.
For localization, this new band structure creates opportunities for resource optimization and robustness, as multi-band diversity provides additional degrees of freedom for balancing accuracy, reliability, and SNR constraints.

Traditionally, most 5G and mmWave localization methods rely on single-band operation and far-field assumptions \cite{shahmansoori2018position,lota2022mmwave,ruble2018wireless,nazari2023mmwave}.
Motivated by FR3 and future heterogeneous-spectrum deployments, this work investigates localization limits in a more general multi-band and near-field setting.

\subsection{Related Works}

A wide range of MIMO-based localization methods have been explored in the literature \cite{Stylianos,Nasir,shahmansoori2018position,lota2022mmwave,ruble2018wireless,nazari2023mmwave,lin2018indoor,chen2022joint,garcia2017direct,taponecco2011joint,Yang,kazaz2022delay,Raviv}.
Common approaches rely on estimating geometric parameters such as AoA, ToA, or TDoA from channel state information (CSI), typically using pilot-based signaling schemes.

Localization in mmWave and massive-MIMO systems has received significant attention due to the availability of large bandwidth and high angular resolution.
Works such as \cite{shahmansoori2018position} derive CRBs for position and rotation estimation under far-field assumptions and single-band operation.
Extensions to full 6D localization have been studied in \cite{nazari2023mmwave}, using non-convex optimization frameworks and providing CRBs for the estimated parameters.
Vehicular applications leveraging TDoA, \ac{FDoA}, or hybrid beamforming have been investigated in \cite{lota2022mmwave,Henk5G}.
Two-step ToA/AoA approaches under LoS and NLoS conditions are considered in \cite{ruble2018wireless}, though these methods require strict synchronization between transmitter and receiver.
Another relevant aspect is the use of distributed MIMO systems for localization, \cite{garcia2017direct,Vladimir}. In \cite{garcia2017direct}, the authors proposed a direct localization method by jointly processing the observations of different base stations.

Systems based on \ac{RSSI} or hybrid RSSI-AoA approaches provide alternatives when propagation models can accurately capture path loss.
Examples include \cite{lin2018indoor,ThomasRSSI,MarcoRSSI}, where hybrid angle–power models require specific array structures or environmental regularity.
Such approaches are often limited by the unpredictable nature of path loss in practical environments.

Multi-band localization has been explored in \cite{kazaz2022delay, Raviv}. In \cite{kazaz2022delay}, the authors focused on a distributed MIMO system where several known anchors perform TDoA-based processing. A super-resolution delay estimation method was derived, as well as CRB expressions. In \cite{Raviv}, the authors focused on incoherent multi-band processing by proposing an FR3 beamformer that capitalizes on the different analog front-ends and spectral behavior associated with each sub-band.

Localization algorithms generally fall into two categories: direct and intermediate-parameter approaches.
Direct approaches process all received signals jointly to infer the position without explicitly estimating intermediate parameters such as delay or AoA \cite{garcia2017direct,Vladimir,Yang}.
In contrast, intermediate-parameter methods first estimate quantities that admit a sparse or harmonic structure (e.g., delay, AoA, Doppler), typically using tools such as over-sampled \ac{FFT}, ESPRIT, \ac{GAMP}, \ac{OMP}, or RIMAX \cite{Henk5GPart1,kazaz2022delay,RimaxThesis,Bomfin_UW,Bomfin_DMC_globecom}.
The resulting geometric parameters are then mapped to range or position through appropriate transformation functions.

\subsection{This Work}\label{subsec:intro_this_work}

Despite substantial recent progress spanning 6D localization \cite{nazari2023mmwave}, estimation in multipath environments \cite{garcia2017direct,ruble2018wireless}, and a broad body of CRB analyses \cite{shahmansoori2018position,nazari2023mmwave,RimaxThesis}, we identify fundamental gaps even in the basic single-path 2D (azimuth) \ac{SIMO} localization model. In particular, to the best of our knowledge, prior work does not provide an explicit condition on the \ac{SNR} under which the distance CRB is attainable. This threshold characterization is critical for system design, since modest changes in waveform, bandwidth, or array geometry can otherwise necessitate extensive new simulation campaigns. In addition, many existing CRB derivations are presented only in compact matrix form: while convenient for numerical evaluation, they offer limited analytical insight into how key system parameters drive localization accuracy. 

Finally, most localization frameworks assume single-band operation, leaving the benefits and design trade-offs of multi-band processing largely unexplored. 
Unlike \cite{kazaz2022delay} that develops TDoA-based processing for distributed MIMO, in this work, we consider a single receiver at near-field with an antenna array.
And our work differs from \cite{Raviv} by considering coherent processing, typical in pilot-based processing schemes.

Motivated by these gaps, this work develops a comprehensive analytical framework that: (i) derives closed-form AoA–distance CRBs for \ac{ULA} systems; (ii) proposes an asymptotic \ac{ML} estimator for joint parameter estimation; and (iii) obtains closed-form \ac{TSNR} conditions that quantify the SNR thresholds separating unreliable, “off-the-chart” behavior from CRB-approaching performance. The framework further extends naturally to multi-band operation with independently parameterized sub-bands, yielding a unified treatment of performance and design trade-offs across bands. Finally, because the underlying model is readily extensible to 3D localization, distributed (cell-free) MIMO architectures, and multipath channels, the proposed analysis provides a versatile theoretical tool for guiding the design of a broad class of localization systems.

The main contributions of this paper are as follows:
\begin{itemize}
\item We introduce a multi-band \ac{SIMO} TDoA-based localization model that accommodates independently parameterized sub-bands, allowing arbitrary sub-carrier spacing and number of sub-carriers within each sub-band.
\item We derive a general, matrix-form CRB expression applicable to arbitrary antenna array geometries.
\item Specializing to \acp{ULA}, we obtain tractable approximate closed-form \acp{CRB} for joint AoA–distance estimation, providing insightful expressions with respect to system parameters.
\item We develop asymptotic \ac{ML} benchmark estimators for both single-band and multi-band operation, and show that they attain the corresponding \acp{CRB} in the high-\ac{SNR} regime.
\item We derive closed-form single-band and multi-band \ac{TSNR} expressions for TDoA and distance estimation, characterizing the SNR thresholds for the onset of CRB-approaching performance.
\end{itemize}

\subsection{Notation}
Scalars, column vectors, and matrices are denoted by lower-case italic, lower-case bold, and upper-case bold letters, respectively (e.g., $z$, $\mathbf{z}$, $\mathbf{Z}$). The $m$th element of $\mathbf{z}$ is $z_m$, and an $M$-dimensional vector is written as $\mathbf{z} = \{z_m\}_{m=0}^{M-1}$. The operators $\mathbf{Z}^{\rm T}$ and $\mathbf{Z}^{\rm H}$ denote the transpose and conjugate transpose of $\mathbf{Z}$.
The notation $\mathbf{Z} = \mathrm{diag}(\mathbf{z})$ constructs a diagonal matrix with diagonal entries taken from $\mathbf{z}$.
$\| \cdot \| $ is the Euclidean norm of a vector.
$\Re(\cdot)$ extracts the real part of units input matrix.
The identity matrix of size $M$ is $\mathbf{I}_M$, while $\mathbf{1}_M$ is the $M$-dimensional all-ones column vector. The zero vector and matrix of sizes $M$ and $M \times N$ are denoted by $\mathbf{0}_M$ and $\mathbf{0}_{M,N}$, respectively.
The Khatri–Rao product of matrices $\mathbf{B}$ and $\mathbf{A}$ is defined as
$
	\mathbf{B} \Diamond \mathbf{A} = 
	\begin{bmatrix}
		\mathbf{b}_0 \otimes \mathbf{a}_0\!\! & \!\!
		\mathbf{b}_1 \otimes \mathbf{a}_1 \!\!& \!\!
		\cdots\!\! & \!\!
		\mathbf{b}_{M-1} \otimes \mathbf{a}_{M-1}
	\end{bmatrix}
$,
where $\mathbf{a}_i$ and $\mathbf{b}_i$ denote the $i$th columns of $\mathbf{A}$ and $\mathbf{B}$, respectively, 
and $\otimes$ represents the Kronecker product.
For a given vector model $\mathbf{z}(\mathbf{v})$ as a function of the parameter vector $\mathbf{v}$, the partial derivative is written as $ \mathbf{z}^{(\mathbf{v})} = \partial \mathbf{z} / \partial \mathbf{v}^{\rm T}$.
When convenient, we use the Jacobian notation $\mathbf{D}(\mathbf{v}) = \mathbf{z}^{(\mathbf{v})}$, yielding the \ac{FIM} and \ac{CRB} of the parameter vector $\mathbf{v}$, $\mathbf{J}(\mathbf{v})
= 2/\sigma^2 
\Re\!\left( \mathbf{D}(\mathbf{v})^{\rm H} 
\mathbf{D}(\mathbf{v}) \right)$ and $\text{CRB}_{\mathbf{v}} = \mathbf{J}(\mathbf{v})^{-1}$, respectively, for the system under \ac{AWGN} with noise power $\sigma^2$.

\subsection{Organization}
The remainder of this paper is organized as follows.
In Section~\ref{sec:system_model_single_band}, we define the single-band TDoA and AoA-Distance system models.
In Section~\ref{sec:crb_R_phi}, we derive the Jacobian matrix and matrix-form CRB expressions for AoA and Distance. 
In Section~\ref{sec:crb_R_phi_approx}, we derive a closed-form expression for the TDoA CRB, the intermediate AoA-Distance CRB expression for arbitrary array shapes, and the closed-form CRBs for AoA and distance for ULA.
In Section~\ref{sec:estimator_R_phi}, we derive the Algorithm for AoA and distance estimation, as well as the distance WNSR.
In Section~\ref{sec:system_model_multi_band}, we generalize the results to multi-band.
In Section~\ref{sec:results}, we provide numerical results to validate the theoretical work.
Finally, Section~\ref{sec:conclusion} concludes the work.

\section{Single-Band System Model}\label{sec:system_model_single_band}
%
%As outlined in Subsection~\ref{subsec:intro_this_work}, 
%In the following, we present the azimuth single-path \ac{SIMO} model.
%In Subsection~\ref{subsec:geometry}, we introduce the geometry-related quantities.
%In Subsection~\ref{subsec:los_simo_model}, we define the signal system.
%Lastly, in Subsection~\ref{subsec:aoa_distance_model}, we define the functional relation between the TDoA and AoA-Distance models.
%In Subsection~\ref{subsec:geometry}, we introduce the quantities related to the system geometry such as distance, AoA, antenna elements delays and TDoAs.
%In Subsection~\ref{subsec:los_simo_model}, we define the signal system including carrier frequency, sub-carrier spacing and number of sub-carriers.
%Lastly, in Subsection~\ref{subsec:aoa_distance_model}, we define the functional relation between the TDoA and AoA-Distance models.
%
\subsection{Enviromental Model}\label{subsec:geometry}
Consider the system model illustrated in Fig.~\ref{fig:system_model}. 
The receive \ac{ULA} consists of $M$ antenna elements uniformly distributed along the $\mathrm{x}$-axis, with an inter-element spacing of $\delta_{\mathrm{x}}$. 
The Cartesian coordinates of the $m$th antenna element are denoted by $(x_m, y_m)$. 
The transmitter position is given by the coordinates $(x_{\mathrm{T}}, y_{\mathrm{T}})$, where $x_{\mathrm{T}} = R \cos \phi$ and $y_{\mathrm{T}} = R \sin \phi \ge 0$ are dependent upon the \ac{AoA}, $\phi$, and distance, $R$.
The AoA is defined within the interval $[0, \pi]$ under the assumption $y_{\mathrm{T}}\ge 0$, which eliminates the ambiguity between positions $\pm y_{\mathrm{T}}$. 
The distance between the transmitter and the $m$th receive antenna is $d_m = \sqrt{(R \cos \phi - x_m)^2 + (R \sin \phi - y_m)^2}$
%\begin{equation}
%	d_m = \sqrt{(R \cos \phi - x_m)^2 + (R \sin \phi - y_m)^2},
%\end{equation}
%
with corresponding propagation delay $	\tau_m = d_m/c$,
%
%\begin{equation}
%	\tau_m = d_m/c,
%\end{equation}
%
where $c$ denotes the speed of light.

\begin{figure}[t!]
	\tikzset{every picture/.style={line width=0.75pt}} %set default line width to 0.75pt        

\begin{tikzpicture}[x=0.75pt,y=0.75pt,yscale=-1,xscale=1]
	%uncomment if require: \path (0,306); %set diagram left start at 0, and has height of 306
	
	%Straight Lines [id:da6125803365037646] 
	\draw [line width=0.75]  [dash pattern={on 4.5pt off 4.5pt}]  (172.57,103.91) -- (261.07,195.41) ;
	\draw [shift={(261.07,195.41)}, rotate = 225.95] [color={rgb, 255:red, 0; green, 0; blue, 0 }  ][line width=0.75]    (0,7.83) -- (0,-7.83)   ;
	\draw [shift={(172.57,103.91)}, rotate = 225.95] [color={rgb, 255:red, 0; green, 0; blue, 0 }  ][line width=0.75]    (0,7.83) -- (0,-7.83)   ;
	%Shape: Arc [id:dp2396918311299261] 
	\draw  [draw opacity=0] (243.97,176.75) .. controls (246.46,175.89) and (249.14,175.42) .. (251.93,175.4) .. controls (264.55,175.34) and (274.93,184.68) .. (275.92,196.59) -- (252.04,198.45) -- cycle ; \draw    (243.97,176.75) .. controls (246.46,175.89) and (249.14,175.42) .. (251.93,175.4) .. controls (264.55,175.34) and (274.93,184.68) .. (275.92,196.59) ;  
	%Straight Lines [id:da19717676801627337] 
	\draw [line width=0.75]  [dash pattern={on 4.5pt off 4.5pt}]  (231,187) -- (260.33,186.75) ;
	\draw [shift={(260.33,186.75)}, rotate = 179.51] [color={rgb, 255:red, 0; green, 0; blue, 0 }  ][line width=0.75]    (0,3.91) -- (0,-3.91)   ;
	\draw [shift={(231,187)}, rotate = 179.51] [color={rgb, 255:red, 0; green, 0; blue, 0 }  ][line width=0.75]    (0,3.91) -- (0,-3.91)   ;
	%Shape: Circle [id:dp4010771151150866] 
	\draw  [fill={rgb, 255:red, 0; green, 0; blue, 0 }  ,fill opacity=1 ] (257.43,195.5) .. controls (257.43,193.57) and (259,192) .. (260.93,192) .. controls (262.86,192) and (264.43,193.57) .. (264.43,195.5) .. controls (264.43,197.43) and (262.86,199) .. (260.93,199) .. controls (259,199) and (257.43,197.43) .. (257.43,195.5) -- cycle ;
	%Shape: Circle [id:dp22821974777963971] 
	\draw  [fill={rgb, 255:red, 0; green, 0; blue, 0 }  ,fill opacity=1 ] (286.74,195.5) .. controls (286.74,193.57) and (288.31,192) .. (290.24,192) .. controls (292.17,192) and (293.74,193.57) .. (293.74,195.5) .. controls (293.74,197.43) and (292.17,199) .. (290.24,199) .. controls (288.31,199) and (286.74,197.43) .. (286.74,195.5) -- cycle ;
	%Shape: Circle [id:dp010021492054048786] 
	\draw  [fill={rgb, 255:red, 0; green, 0; blue, 0 }  ,fill opacity=1 ] (324.36,195.5) .. controls (324.36,193.57) and (325.92,192) .. (327.86,192) .. controls (329.79,192) and (331.36,193.57) .. (331.36,195.5) .. controls (331.36,197.43) and (329.79,199) .. (327.86,199) .. controls (325.92,199) and (324.36,197.43) .. (324.36,195.5) -- cycle ;
	%Shape: Circle [id:dp43306108938934984] 
	\draw  [fill={rgb, 255:red, 0; green, 0; blue, 0 }  ,fill opacity=1 ] (228.12,195.5) .. controls (228.12,193.57) and (229.69,192) .. (231.62,192) .. controls (233.55,192) and (235.12,193.57) .. (235.12,195.5) .. controls (235.12,197.43) and (233.55,199) .. (231.62,199) .. controls (229.69,199) and (228.12,197.43) .. (228.12,195.5) -- cycle ;
	%Shape: Axis 2D [id:dp7127056371423113] 
	\draw  (128.86,185.31) -- (146.57,185.31)(130.63,169.41) -- (130.63,187.07) (139.57,180.31) -- (146.57,185.31) -- (139.57,190.31) (125.63,176.41) -- (130.63,169.41) -- (135.63,176.41)  ;
	%Shape: Circle [id:dp25066327393842514] 
	\draw  [fill={rgb, 255:red, 0; green, 0; blue, 0 }  ,fill opacity=1 ] (169.36,104.07) .. controls (169.36,102.14) and (170.92,100.57) .. (172.86,100.57) .. controls (174.79,100.57) and (176.36,102.14) .. (176.36,104.07) .. controls (176.36,106.01) and (174.79,107.57) .. (172.86,107.57) .. controls (170.92,107.57) and (169.36,106.01) .. (169.36,104.07) -- cycle ;
	%Shape: Circle [id:dp9725330610907481] 
	\draw  [fill={rgb, 255:red, 0; green, 0; blue, 0 }  ,fill opacity=1 ] (186.5,195.5) .. controls (186.5,193.57) and (188.07,192) .. (190,192) .. controls (191.93,192) and (193.5,193.57) .. (193.5,195.5) .. controls (193.5,197.43) and (191.93,199) .. (190,199) .. controls (188.07,199) and (186.5,197.43) .. (186.5,195.5) -- cycle ;
	
	% Text Node
	\draw (345,187) node [anchor=north west][inner sep=0.75pt]  [font=\normalsize] [align=left] {Receive Array};
	% Text Node
	\draw (221.67,133.4) node [anchor=north west][inner sep=0.75pt]    {$R$};
	% Text Node
	\draw (183.83,97.33) node [anchor=north west][inner sep=0.75pt]  [font=\small] [align=left] {Trasmit Antenna at $(x_{\rm T},y_{\rm T})= (R \cos \phi, R \sin \phi)$};
	% Text Node
	\draw (266,160.07) node [anchor=north west][inner sep=0.75pt]    {$\phi $};
	% Text Node
	\draw (298.67,192) node [anchor=north west][inner sep=0.75pt]  [font=\normalsize]  {$\cdots $};
	% Text Node
	\draw (128.47,201.4) node [anchor=north west][inner sep=0.75pt]  [font=\normalsize]  {$m=-M/2$};
	% Text Node
	\draw (218.27,201.4) node [anchor=north west][inner sep=0.75pt]  [font=\normalsize]  {$-1$};
	% Text Node
	\draw (254.87,201.4) node [anchor=north west][inner sep=0.75pt]  [font=\normalsize]  {$0$};
	% Text Node
	\draw (314.27,201.4) node [anchor=north west][inner sep=0.75pt]  [font=\normalsize]  {$M/2-1$};
	% Text Node
	\draw (236.93,189.33) node [anchor=north west][inner sep=0.75pt]    {$\delta _{x}$};
	% Text Node
	\draw (149.6,176.39) node [anchor=north west][inner sep=0.75pt]    {$x$};
	% Text Node
	\draw (126.67,147.86) node [anchor=north west][inner sep=0.75pt]    {$y$};
	% Text Node
	\draw (283.87,201.4) node [anchor=north west][inner sep=0.75pt]  [font=\normalsize]  {$1$};
	% Text Node
	\draw (199.67,192) node [anchor=north west][inner sep=0.75pt]  [font=\normalsize]  {$\cdots $};

\end{tikzpicture}
	\vspace{-1.0cm}
	\caption{System model. Single input-multiple output (SIMO). The antenna array is distributed along the $\rm x$ axis.}
	\label{fig:system_model}
    \vspace{-0.5cm}
\end{figure}
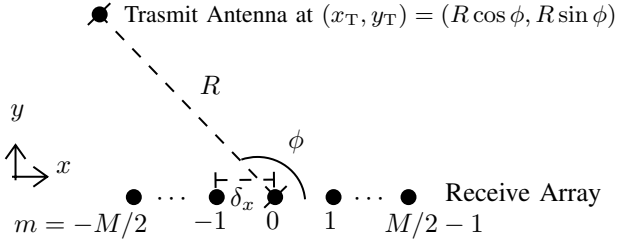

%In this work, we tackle the localization problem via \ac{TDoA} processing that does not require absolute time synchronization between the transmitter and receiver, which is hard to obtain in practice.
%\subsection{Time Difference of Arrival (TDoA) Model}\label{subsec:tdoa_model}
%
Since absolute delay estimation requires TX–RX time synchronization, we instead work with TDoAs, which eliminate the unknown transmit time. To overcome this limitation, we tackle the localization problem via \ac{TDoA} processing.
The TDoA parameter is $\delta_m = \tau_m - \tau_0$,
%
%\begin{equation}
%	\delta_m = \tau_m - \tau_0,
%	\label{eq:tdoa_m}
%\end{equation}
%
which represents the relative delay between the $m$th antenna and the reference ($0$th) antenna. 
The TDoA parameter vector is 
\begin{equation}
\boldsymbol{\delta} = \left[\delta_1 \,\, \cdots \,\, \delta_{M-1}\right]^{\mathrm{T}} \in \mathbb{R}^{M-1}.
\label{eq:delta}
\end{equation}

\subsection{LoS SIMO Signal Model}\label{subsec:los_simo_model}
We focus on a single dominant LoS component and treat residual multipath as negligible (or absorbed into noise), which is appropriate when secondary paths are sufficiently weaker and/or separable in delay/angle. Also, we consider a \ac{SIMO} system transmitting a constant-amplitude signal in the frequency domain. 
Let
\begin{equation}
\begin{aligned}
	&\mathbf{a}(\tau) = \{\exp(-j 2 \pi \tau f_0 (n - (N-1)/2))\}_{n=0}^{N-1}
    \\ 
     \text{and } &b(\tau) = \exp(-j 2 \pi \tau f_{\mathrm{c}})	
    \end{aligned}\label{eq:ab}
\end{equation}
%and
%
%\begin{equation}
%	b(\tau) = \exp(-j 2 \pi \tau f_{\mathrm{c}})
%	\label{eq:b}
%\end{equation}
%
denote the frequency-domain basis function and the systematic phase rotation, respectively. The quantities $N$ and $f_0$ represent the subcarrier spacing and the number of subcarriers, respectively. % Note that we use a baseband-centered subcarrier index $n-(N-1)/2$.
The received signal at the $m$-th antenna element is then given by
\begin{equation}
	\mathbf{s}_m = \mathbf{a}(\tau_0 + \delta_m)\, b(\delta_m)\, \gamma,
	\label{eq:s_m}
\end{equation}
where $\mathbf{s}_m\in\mathbb{C}^{N}$ and $\gamma$ is the complex signal gain, whose phase is independent of $m$. 
In~\eqref{eq:s_m}, note that $b(\delta_m)$ does not include $\tau_0$, reflecting the fact that it is already absorbed into $\gamma = \tilde{\gamma}\, b(\tau_0)$. 
This modeling choice avoids the non-identifiability between the phase of $\tilde{\gamma}$ and $b(\tau_0)$. 
Let
\begin{equation}
		\boldsymbol{\theta}  = [\, \tau_0 \,\, \delta_1 \,\, \cdots \,\, \delta_{M-1} \,\, \gamma^{\mathrm{r}} \,\, \gamma^{\mathrm{i}} \,]^{\mathrm{T}} \in \mathbb{R}^{M+2}
\end{equation}
be the parameter vector containing the TDoAs and channel gain.
A convenient and compact linear representation of the system is given by
\begin{equation}
	\begin{aligned}
		\mathbf{s}(\boldsymbol{\theta}) 
		= [ \mathbf{s}_0^{\mathrm{T}} \,\,  \cdots \,\, \mathbf{s}_{M-1}^{\mathrm{T}} ]^{\mathrm{T}} = \big( \mathbf{B}(\boldsymbol{\delta}_0) \Diamond \mathbf{A}(\tau_0 + \boldsymbol{\delta}_0^{\mathrm{T}}) \big) \mathbf{1}_M \, \gamma
	\end{aligned}
	\label{eq:s_theta}
\end{equation}
where
\begin{equation}
	\boldsymbol{\delta}_0 = [0 \,\, \boldsymbol{\delta}^{\mathrm{T}}]^{\mathrm{T}} \in \mathbb{R}^{M}
\end{equation}
is the $M$-size TDoA vector with the first element $(m=0)$ being zero as the TDoA is computed with respect to the first antenna.
The frequency-domain basis function matrix of size $\mathbf{A}(\tau_0 + \boldsymbol{\delta}_0^{\mathrm{T}})\in \mathbb{C}^{N \times  M}$ is
\begin{equation}
	\mathbf{A}(\tau_0 + \boldsymbol{\delta}_0^{\mathrm{T}}) = \left[\mathbf{a}(\tau_0) \,\, \mathbf{a}(\tau_0  +  \delta_1) \,\, \cdots \,\,  \mathbf{a}(\tau_0  +  \delta_{M  -1})\right] 
\end{equation}
and stacks the vectors $\mathbf{a}(\cdot)$ in \eqref{eq:ab} along columns. The diagonal matrix
\begin{equation}
\mathbf{B}(\boldsymbol{\delta}_0)= \mathrm{diag}\big(\{b(\delta_m)\}_{m=0}^{M-1}\big) \in \mathbb{C}^{M \times M}
\end{equation}
is constructed with the elements of ${b}(\cdot)$ in \eqref{eq:ab}. 
Since $\mathbf{B}(\boldsymbol{\delta}_0)$ is diagonal, the Khatri--Rao product compactly represents the stacking of the per-antenna vectors $\mathbf{a}(\tau_0+\delta_m)b(\delta_m)$ in a single expression. This structure is used for convenience as it allows us to compactly write Jacobian and CRB expressions for different parameters of interest.

The received signal under \ac{AWGN} is given by
\begin{equation}
	\mathbf{y} = \mathbf{s}(\boldsymbol{\theta}) + \mathbf{w},
	\label{eq:y}
\end{equation}
\noindent where $\mathbf{y}\in\mathbb{C}^{MN}$, $\mathbf{w}\in\mathbb{C}^{MN}$, $\mathbf{w}\sim \mathcal{CN}(\mathbf{0},\sigma^2\mathbf{I})$ and $\sigma^2$ is the noise power.

\subsection{AoA-Distance Model}\label{subsec:aoa_distance_model}
The TDoA is related to the parameters $R$ and $\phi$ as
\begin{equation}
	\delta_m 
	= \frac{1}{c} 
	\sqrt{(R \cos \phi - x_m)^2 + (R \sin \phi - y_m)^2} 
	- \frac{R}{c}.
	\label{eq:delta_R_pi}
\end{equation}
The TDoA parameter vector $\boldsymbol{\delta}$ in~\eqref{eq:delta}, 
and consequently $\boldsymbol{\theta}$, 
can be expressed as a function of $[R \,\, \phi]^{\mathrm{T}}$ through~\eqref{eq:delta_R_pi}.
For completeness, we define an alternative parameter vector
\begin{equation}
	\boldsymbol{\upsilon} = 
	[\, \tau_0 \,\, \phi \,\, R \,\, \gamma^{\mathrm{r}} \,\, \gamma^{\mathrm{i}} \,]^{\mathrm{T}}.
\end{equation}
When convenient, the TDoA dependence on $R$ and $\phi$ according to \eqref{eq:delta_R_pi} is made explicit using  $\boldsymbol{\theta}(\boldsymbol{\upsilon})$. Note that we use $\boldsymbol{\theta}$ for the TDoA parameterization and $\boldsymbol{\upsilon}$ for the physically meaningful $(\phi,R)$ parameterization.

%Lastly, we note that there is no relation between $\tau_0$ and $R$ due to the absence of absolute time synchronization between the transmitter and receiver.
%
\section{Matrix-form CRB of $\phi$ and $R$}\label{sec:crb_R_phi}
The goal of this section is to derive the matrix-form CRB for the parameters $\phi$ and $R$ under the model presented in Section~\ref{sec:system_model_single_band}. 
Our approach begins by formulating the Jacobian matrix of $\boldsymbol{\theta}$ corresponding to the TDoA-based model in Subsection~\ref{subsec:D_theta}. 
We then demonstrate how this formulation is transformed into the $\boldsymbol{\upsilon}$ (AoA-Distance) representation in Subsections~\ref{subsec:D_upsilon} and~\ref{subsec:crb_phi_R_matrix}.

\subsection{Jacobian of $\boldsymbol{\theta} = [ \tau_0 \,\, \delta_1 \,\, \cdots \,\, \delta_{M-1} \,\, \gamma^{\rm r} \,\, \gamma^{\rm i}]^{\rm T}$}\label{subsec:D_theta}
From the model in \eqref{eq:s_theta}, the Jacobian of $\boldsymbol{\theta}$ is given by
\begin{equation}
	\mathbf{D}(\boldsymbol{\theta}) 
	= \frac{\partial \mathbf{s}(\boldsymbol{\theta})}{\partial \boldsymbol{\theta}^{\rm T}} 
	= \big[\, \mathbf{S}^{(\tau_0)} \,\, \mathbf{S}^{(\boldsymbol{\delta})} \,\, \mathbf{S}^{(\gamma)} \,\big],
	\label{eq:D_theta}
\end{equation}

\noindent where $\mathbf{D}(\boldsymbol{\theta})\in\mathbb{C}^{MN\times(M+2)}$, 
$\mathbf{S}^{(\tau_0)}\in\mathbb{C}^{MN\times 1}$, 
$\mathbf{S}^{(\boldsymbol{\delta})}\in\mathbb{C}^{MN\times(M-1)}$, and 
$\mathbf{S}^{(\gamma)}\in\mathbb{C}^{MN\times 2}$, and  the partial derivatives with respect to $\tau_0$, $\boldsymbol{\delta}$, and $[\gamma^{\rm r} \,\, \gamma^{\rm i}]^{\rm T}$ are expressed as
\begin{equation}
	\begin{aligned}
		\mathbf{S}^{(\tau_0)} 
		&\!=\! \frac{\partial \mathbf{s}(\boldsymbol{\theta})}{\partial \tau_0} 
		\!=\! \big( {\mathbf{B}}(\boldsymbol{\delta}_0) 
		\lozenge 
		(\boldsymbol{\Xi}(f_0, 0)\mathbf{A}(\tau_0 + \boldsymbol{\delta}_0^{\rm T})) \big) 
		\mathbf{1}_M\, \gamma, \\[3pt]
		\mathbf{S}^{(\boldsymbol{\delta})} 
		&\!=\! \frac{\partial \mathbf{s}(\boldsymbol{\theta})}{\partial \boldsymbol{\delta}^{\rm T}} 
		\!= \!\big( {\mathbf{B}}'(\boldsymbol{\delta}) 
		\lozenge 
		(\boldsymbol{\Xi}(f_0, f_{\rm c})\mathbf{A}(\tau_0 + \boldsymbol{\delta}^{\rm T})) \big)\gamma, \\[3pt]
		\text{and }\mathbf{S}^{(\gamma)} 
		&\!=\! \frac{\partial \mathbf{s}(\boldsymbol{\theta})}{\partial [\gamma^{\rm r} \,\, \gamma^{\rm i}]} 
		\!=\! \big( {\mathbf{B}}(\boldsymbol{\delta}_0) 
		\lozenge 
		\mathbf{A}(\tau_0 + \boldsymbol{\delta}_0^{\rm T}) \big) 
		[\,\mathbf{1}_M \,\, j\mathbf{1}_M\,],
	\end{aligned}
	\label{eq:S_delta_derivatives}
\end{equation}
where the following auxiliary variables are defined for convenience: $\boldsymbol{\Xi}(f_0, f_{\rm c}) 
= -j 2\pi (\grave{\mathbf{N}} f_0 + \mathbf{I} f_{\rm c}) 
\in \mathbb{C}^{N \times N}$, ${\mathbf{B}}'(\boldsymbol{\delta}) 
= 
\begin{bmatrix} 
	\mathbf{0}_{M-1} & 
	{\mathbf{B}}(\boldsymbol{\delta})
\end{bmatrix}^{\rm T}
\in \mathbb{C}^{M \times (M-1)}$, and $\grave{\mathbf{N}} 
= {\rm diag}\!\left( [0\,\, 1 \cdots N \!-\!1]^{\rm T} - \tfrac{N-1}{2} \right) 
\in \mathbb{Q}^{N \times N}$.
%
%\begin{equation}
%	\begin{aligned}
%		\grave{\mathbf{N}} 
%		&= {\rm diag}\!\left( [0\,\, 1 \cdots N \!-\!1]^{\rm T} - \tfrac{N-1}{2} \right) 
%		\in \mathbb{Q}^{N \times N}, \\[3pt]
%		\boldsymbol{\Xi}(f_0, f_{\rm c}) 
%		&= -j 2\pi (\grave{\mathbf{N}} f_0 + \mathbf{I} f_{\rm c}) 
%		\in \mathbb{C}^{N \times N}, \\[3pt]
%		{\mathbf{B}}'(\boldsymbol{\delta}) 
%		&= 
%		\begin{bmatrix} 
%			\mathbf{0}_{M-1}^{\rm T} \\[2pt]
%			{\mathbf{B}}(\boldsymbol{\delta})
%		\end{bmatrix} 
%		\in \mathbb{C}^{M \times (M-1)}.
%	\end{aligned}
%%	\label{eq:aux_var2}
%\end{equation}
%

The diagonal matrix $\boldsymbol{\Xi}(f_0, f_{\rm c})$ is obtained using the product rule and the complex exponential definitions of $a(\cdot)$ and $b(\cdot)$ in \eqref{eq:ab}. 
The matrix ${\mathbf{B}}'(\boldsymbol{\delta})$ is introduced as the zero-padded version of ${\mathbf{B}}(\boldsymbol{\delta})$ (i.e., excluding $\tau_0$), which ensures that $\mathbf{S}^{(\boldsymbol{\delta})}\in\mathbb{C}^{MN\times(M-1)}$ and that the first $N$ entries (corresponding to the reference antenna $m=0$) are zero.

\subsection{Jacobian of ${\boldsymbol{\upsilon}}$}\label{subsec:D_upsilon}
Due to the relation between $[\phi \,\, R]^{\mathrm{T}}$ and $\delta_m$ defined in \eqref{eq:delta_R_pi}, the Jacobian of $\boldsymbol{\upsilon}$ can be expressed from $\mathbf{D}(\boldsymbol{\theta})$ in \eqref{eq:D_theta} using the chain rule as
\begin{equation}
	\mathbf{D}(\boldsymbol{\upsilon}) 
	=  \mathbf{D}(\boldsymbol{\theta}(\boldsymbol{\upsilon})) \boldsymbol{\theta}^{(\boldsymbol{\upsilon})}
    = [\, \mathbf{S}^{(\tau_0)} \,\, \mathbf{S}^{(\phi)} \,\, \mathbf{S}^{(R)} \,\, \mathbf{S}^{(\gamma)} \,],
	\label{eq:D_upsilon2}
\end{equation}
where $\boldsymbol{\theta}^{(\boldsymbol{\upsilon})} \in \mathbb{R}^{(M+2) \times 5}$ is given by
\begin{equation}
	\boldsymbol{\theta}^{(\boldsymbol{\upsilon})}
	\!	= 	\!\frac{\partial \boldsymbol{\theta}(\boldsymbol{\upsilon})}{\partial \boldsymbol{\upsilon}^{\rm T}}
		\!= 	\!
	\begin{bmatrix}
		\begin{bmatrix}
			1 \\[2pt] \mathbf{0}_{M-1}
		\end{bmatrix} 
	\!	& \!
		\begin{bmatrix}
			0 \\[2pt] \boldsymbol{\delta}^{(\phi)} 
		\end{bmatrix} 		
	\!	& 	\!	
		\begin{bmatrix}
			0 \\[2pt] \boldsymbol{\delta}^{(R)} 
		\end{bmatrix} 	
	\!	&\! \mathbf{0}_{M,2} \\[2pt]
		\mathbf{0}_2 & \mathbf{0}_2 & \mathbf{0}_2 & \mathbf{I}_2
	\end{bmatrix}.
\end{equation}
The partial-derivative vectors are defined as 
$\boldsymbol{\delta}^{(R)} = [\delta_1^{(R)} \,\, \cdots \,\, \delta_{M-1}^{(R)}]^{\rm T}\in\mathbb{R}^{M-1}$
and 
$\boldsymbol{\delta}^{(\phi)} = [\delta_1^{(\phi)} \,\, \cdots \,\, \delta_{M-1}^{(\phi)}]^{\rm T}\in\mathbb{R}^{M-1}$, 
%
%\begin{equation}
%	\begin{aligned}
%		\boldsymbol{\delta}^{(R)} 
%		&= \frac{\partial \boldsymbol{\delta}}{\partial R} 
%		= [\delta_1^{(R)} \,\, \cdots \,\, \delta_{M-1}^{(R)}]^{\rm T} 
%		\in \mathbb{R}^{M-1}, \\[3pt]
%		\boldsymbol{\delta}^{(\phi)} 
%		&= \frac{\partial \boldsymbol{\delta}}{\partial \phi} 
%		= [\delta_1^{(\phi)} \,\, \cdots \,\, \delta_{M-1}^{(\phi)}]^{\rm T} 
%		\in \mathbb{R}^{M-1},
%	\end{aligned}
%\end{equation}
%
and $\delta_m^{(R)} = \frac{\partial \delta_m}{\partial R}$ and $\delta_m^{(\phi)} = \frac{\partial \delta_m}{\partial \phi}$ are given by
	\begin{equation}
		\begin{aligned}
			\delta_m^{(R)} 
			&\!	= \!\frac{(R\cos \phi \!-\! x_m)\cos \phi \!+ \!(R\sin\phi \!-\! y_m)\sin\phi }
			{c\sqrt{(R \cos\phi-x_m)^2 \!+\! (R\sin\phi - y_m)^2}} 
		\!	- \! \frac{1}{c} \\[4pt]
		%	\stackrel{R\to \infty}{=} 0, \\[4pt]
			\text{and } \delta_m^{(\phi)} 
			& \!= \!\frac{R (x_m\sin \phi - y_m \cos \phi)} 
			{c\sqrt{(R \cos \phi-x_m)^2 + (R\sin \phi - y_m)^2}} 
			%\stackrel{R\to \infty}{=} \frac{1}{c}(x_m\sin \phi - y_m\cos \phi).
		\end{aligned}
		\label{eq:partial_delta_m}
	\end{equation}
using the TDoA mapping of \eqref{eq:delta_R_pi}. 
%Notice that, in the far-field regime, $\delta_m^{(R)}$ becomes independent of $R$ in the limit $R \to \infty$, implying that the model no longer carries information about $R$, and hence the distance cannot be estimated. 
Since $R$ and $\phi$ affect only $\boldsymbol{\delta}$, $\mathbf{D}(\boldsymbol{\upsilon})$ can be written compactly as the right-hand side of \eqref{eq:D_upsilon2} using the quantities $\mathbf{S}^{(R)} 
		= \frac{\partial \mathbf{s}(\boldsymbol{\theta}(\boldsymbol{\upsilon}))}{\partial R} 
		= \mathbf{S}^{(\boldsymbol{\delta})}\boldsymbol{\delta}^{(R)}$ and $\mathbf{S}^{(\phi)} 
		= \frac{\partial \mathbf{s}(\boldsymbol{\theta}(\boldsymbol{\upsilon}))}{\partial \phi} 
		= \mathbf{S}^{(\boldsymbol{\delta})}\boldsymbol{\delta}^{(\phi)}$, where the chain rule was used.
%
\begin{comment}
\begin{equation}
	\mathbf{D}(\boldsymbol{\upsilon}) 
	= [\, \mathbf{S}^{(\tau_0)} \,\, \mathbf{S}^{(R)} \,\, \mathbf{S}^{(\phi)} \,\, \mathbf{S}^{(\gamma)} \,],
	%\label{eq:D_upsilon}
\end{equation}
%
where
%
\begin{equation}
	\begin{aligned}
		\mathbf{S}^{(R)} 
		&= \frac{\partial \mathbf{s}(\boldsymbol{\theta}(\boldsymbol{\upsilon}))}{\partial R} 
		= \mathbf{S}^{(\boldsymbol{\delta})}\boldsymbol{\delta}^{(R)} 
		\in \mathbb{C}^{MN}, \\[3pt]
		\mathbf{S}^{(\phi)} 
		&= \frac{\partial \mathbf{s}(\boldsymbol{\theta}(\boldsymbol{\upsilon}))}{\partial \phi} 
		= \mathbf{S}^{(\boldsymbol{\delta})}\boldsymbol{\delta}^{(\phi)} 
		\in \mathbb{C}^{MN}.
	\end{aligned}
	%\label{eq:S_R_phi}
\end{equation}
\end{comment}
%
%To provide some intuition, here we note that the structure of $\mathbf{S}^{(\boldsymbol{\delta})}$ in \eqref{eq:S_delta_derivatives} contains only non zero elements in the appropriate rows associated with the $m$th TDoA.
%This happens ${\mathbf{B}}'(\boldsymbol{\delta})$ has a diagonal structure, therefore its Katri-Rao product with $\mathbf{\boldsymbol{\Xi}}(f_0, f_{\rm c}) \mathbf{A}(\tau_0 + \boldsymbol{\delta}^{\rm T})$ creases a larger matrix with a block-diagonal structure.
%Then, multiplying $\mathbf{S}^{(\boldsymbol{\delta})}$ with $\boldsymbol{\delta}^{(R)}$ or $\boldsymbol{\delta}^{(\phi)}$ guarantees that thet correct non zero rows matches the appropriate $m$th derivative.
%Then the Jacobian of $\boldsymbol{\upsilon}$ in \eqref{eq:D_upsilon2} can be written as 
%

%
\subsection{CRB of $\phi$ and $R$}\label{subsec:crb_phi_R_matrix}
Under circularly symmetric complex AWGN, $\mathbf{w}\sim\mathcal{CN}(\mathbf{0},\sigma^2\mathbf{I})$, the \ac{FIM} for $\boldsymbol{\upsilon}$ is 
\begin{equation}
	\mathbf{J}(\boldsymbol{\upsilon}) 
	= \frac{2}{\sigma^2} 
	\Re\!\left( \mathbf{D}(\boldsymbol{\upsilon})^{\rm H} 
	\mathbf{D}(\boldsymbol{\upsilon}) \right),
	\label{eq:J_upsilon}
\end{equation}
and the \ac{CRB} for $R$ and $\phi$ are obtained as
\begin{equation}
	{\rm CRB}_{\phi} = [\mathbf{J}(\boldsymbol{\upsilon})^{-1}]_{2,2}	\quad \text{and}	\quad 
    {\rm CRB}_{R} = [\mathbf{J}(\boldsymbol{\upsilon})^{-1}]_{3,3}.
	\label{eq:crb_R_phi_no_cf}
\end{equation}
Although the \ac{CRB} in \eqref{eq:crb_R_phi_no_cf} does not provide immediate analytical insight due to its matrix-form definition, it remains practically useful. 
For example, the \ac{CRB} can be computed numerically to benchmark the performance of an estimator. 
Moreover, the Jacobian and \ac{FIM} associated with $\boldsymbol{\theta}$ and $\boldsymbol{\upsilon}$ are employed in the implementation of the estimators of $R$ and $\phi$. 
Finally, since \eqref{eq:delta_R_pi} and \eqref{eq:partial_delta_m} impose no restriction on the array geometry, the expression in \eqref{eq:crb_R_phi_no_cf} is general and applicable to arbitrary array configurations.

\section{Closed-form CRB of $\delta_m$, $R$ and $\phi$}\label{sec:crb_R_phi_approx}
In this section, we provide closed-form CRB expressions for the estimation of $\delta_m$, $R$ and $\phi$.
We start by writing the modified delay model in Subsection \ref{subsec:modified_model}, which has a convenient structure without altering the CRB of the desired variables.
Then, via parameter transformation, we provide the TDoA CRB in Subsection \ref{subsec:CRB_tdoa}.
Using the same approach, we provide an intermediate and closed-form CRB expressions for distance and AoA in Subsection~\ref{subsec:CRBs_phi_R}.
\subsection{Modified Delay Model with $\gamma = \tilde{\gamma}\, b(\tau_0)$ }\label{subsec:modified_model}
%
%The first attempt to deriving the close-form CRB of $R$ and $\phi$ is to resolve the expressions of \eqref{eq:crb_R_phi_no_cf} analytically.
%Upon manipulating this expressions algebraically, we learned the factor preventing the closed-form derivation is the model asymmetry between $\tau_0$ and $\delta_m$ concerning the component $b(\delta_m)$ in \eqref{eq:s_m}.
%We recall that the random channel phase introduced by the delay $\exp(-j 2 \pi \tau_0 f_{\mathrm{c}})$ is absorbed into $\gamma = \tilde{\gamma}\, b(\tau_0)$ because this initial phase is indistinguishable from system random phase $\tilde{\gamma}$.
%This means that the phase rotation due to $f_{\rm c}$ is only capture by the TDoAs, $\delta_m$, and not the delay, $\tau_0$, therefore creating the asymmetry between these two variables.
%This asymmetry is explicitly captured in \eqref{eq:S_delta_derivatives} where $\mathbf{S}^{(\boldsymbol{\delta})}$ has $\mathbf{\boldsymbol{\Xi}}(f_0, f_{\mathrm{c}})$ in its definition, while  $\mathbf{S}^{(\tau_0)}$ has $\mathbf{\boldsymbol{\Xi}}(f_0, 0)$.

It is convenient to modify the model in \eqref{eq:s_m} by writing $\gamma~=~\tilde{\gamma}\, b(\tau_0)$ yielding
\begin{equation}
	\tilde{\mathbf{s}}_m = \mathbf{a}(\tau_0 + \delta_m)\, b(\tau_0 + \delta_m)\,\tilde{\gamma}.
	\label{eq:s_m_tilde}
\end{equation}
\noindent Note that $|b(\tau_0)|=1$, hence $|\gamma|=|\tilde{\gamma}|$.
%Compared to the original model in \eqref{eq:s_m}, we note that $\tilde{\gamma} = {\gamma}b(-\tau_0)$, implying that $\angle\tilde{\gamma} = \angle{\gamma} - 2 \pi \tau_0 f_{\mathrm{c}}$. 
The above model is a transformation of \eqref{eq:s_m} via $\tilde{\gamma} = {\gamma}b(-\tau_0)$, which couples the phase of $\tilde{\gamma}$ with $\tau_0$. 
Since $\tilde{\gamma}=\gamma b(-\tau_0)$ is a one-to-one reparametrization that leaves $\boldsymbol{\delta}$ unchanged, the induced \ac{CRB} for $\boldsymbol{\delta}$ (and therefore for $R$ and $\phi$ via \eqref{eq:delta_R_pi}) is invariant under this transformation. Therefore, the modified model can be used to compute the closed-form \ac{CRB} of $R$ and $\phi$.

%\subsection{FIM of delay parameter $\tilde{\boldsymbol{\theta}}_\tau$}
%
%When using the modified model in \eqref{eq:s_m_tilde}, 
It is more convenient to work with the delay parameters $\tilde{\boldsymbol{\theta}}_\tau = [\, \tau_0 \,\, \tau_1 \,\, \cdots \,\, \tau_{M-1} \,\, \gamma^{\mathrm{r}} \,\, \gamma^{\mathrm{i}} \,]^{\mathrm{T}}$ and obtain the TDoA, distance, and angle expressions through parameter transformation. 
The modified delay-based version of \eqref{eq:s_theta} is
\begin{equation}
	\tilde{\mathbf{s}}(\boldsymbol{\theta}) 
	= \big( {\mathbf{B}}(\boldsymbol{\tau}) \lozenge \mathbf{A}(\boldsymbol{\tau}^{\rm T}) \big) 
	\mathbf{1}_M \tilde{\gamma},
	\label{eq:s_theta_tilde}
\end{equation}
where $\boldsymbol{\tau}=  [\tau_0 \,\, \tau_1 \,\, \cdots \,\, \tau_{M-1} ]^{\rm T}$ is the delay vector. We assume $\mathbf{w}\sim\mathcal{CN}(\mathbf{0},\sigma^2\mathbf{I})$. 
Following the same steps as in Subsection~\ref{subsec:D_theta}, the Jacobian of $\tilde{\boldsymbol{\theta}}_\tau$ is $\mathbf{D}(\tilde{\boldsymbol{\theta}}_\tau) 	= \big[\, {\mathbf{S}}^{(\boldsymbol{\tau})} \,\, \mathbf{S}^{(\tilde{\gamma})} \, \big]$, where ${\mathbf{S}}^{(\boldsymbol{\tau})} 
		= \big( {\mathbf{B}}(\boldsymbol{\tau}) 
		\lozenge 
		(\boldsymbol{\Xi}(f_0, f_c)\, \mathbf{A}(\boldsymbol{\tau}^{\rm T})) \big) 
		\tilde{\gamma}$ 
        and
        $\mathbf{S}^{(\tilde{\gamma})} 
		= \big( {\mathbf{B}}(\boldsymbol{\tau}) 
		\lozenge 
		\mathbf{A}(\boldsymbol{\tau}^{\rm T}) \big) 
		\left[\mathbf{1}_M \,\, j\mathbf{1}_M\right]$.
%
\begin{comment}
\begin{equation}
	\mathbf{D}(\tilde{\boldsymbol{\theta}}_\tau) 
	= \big[\, {\mathbf{S}}^{(\boldsymbol{\tau})} \,\, \mathbf{S}^{(\tilde{\gamma})} \, \big],
\end{equation}
%
where
%
\begin{equation}
	\begin{aligned}
		{\mathbf{S}}^{(\boldsymbol{\tau})} 
		&= \big( {\mathbf{B}}(\boldsymbol{\tau}) 
		\lozenge 
		(\boldsymbol{\Xi}(f_0, f_c)\, \mathbf{a}(\boldsymbol{\tau}^{\rm T})) \big) 
		\tilde{\gamma}, \\[3pt]
		\mathbf{S}^{(\tilde{\gamma})} 
		&= \big( {\mathbf{B}}(\boldsymbol{\tau}) 
		\lozenge 
		\mathbf{A}(\boldsymbol{\tau}^{\rm T}) \big) 
		\left[\mathbf{1}_M \,\, j\mathbf{1}_M\right].
	\end{aligned}
	\label{eq:S_delta_derivatives_delay}
\end{equation}
\end{comment}
%
\begin{proposition}
The inverse \ac{FIM} of $\tilde{\boldsymbol{\theta}}_\tau$ has the structure
\begin{equation}
	\mathbf{J}(\tilde{\boldsymbol{\theta}}_\tau)^{-1} 
	= \frac{\sigma^2}{2} 
	\begin{bmatrix} 
		\mathbf{I}_M \frac{1}{\alpha} 
		+ \frac{\beta}{\alpha^2} \mathbf{1}_M \mathbf{1}_M^{\rm T} 
		& \cdots \\[3pt]
		\cdots & \ddots 
	\end{bmatrix},
	\label{eq:J_delay}
\end{equation}
where 
\begin{equation}
	\alpha 
= |\tilde{\gamma}|^2 (2\pi)^2 N 
\left( \frac{(N^2 - 1) f_0^2}{12} + f_{\rm c}^2 \right),
\label{eq:alpha}
\end{equation}
and $\beta$ is a scalar. 
%
%The variable $\beta$ and the lower-right rows and columns associated with the channel gain are omitted here for brevity.
\end{proposition}
\begin{proof}
	See Appendix~\ref{app:inv_FIM_theta_delay} for details.
\end{proof}
The key feature of \eqref{eq:J_delay} lies in the top-left block of the inverse matrix, which contains the \ac{CRB} associated with the delay parameters. 
This block exhibits a convenient diagonal-plus-rank-one structure that enables a closed-form expression for the \ac{CRB} of the TDoAs, distance, and angle via parameter transformation.
%Compared to \eqref{eq:S_delta_derivatives} (TDoA model), the modified model in \eqref{eq:S_delta_derivatives_delay} (delay model) exhibits a symmetric structure. 
%In the TDoA model, $\tau_0$ acts as a special case, making the closed-form expression of $\mathbf{J}(\boldsymbol{\theta})^{-1}$ more difficult to obtain. 
%In contrast, the delay model allows for the closed-form derivation of $\mathbf{J}(\tilde{\boldsymbol{\theta}}_\tau)^{-1}$ in \eqref{eq:J_delay}. 
%As discussed in Subsection~\ref{subsec:modified_model}, the modified model does not alter the \ac{CRB} of the TDoAs, implying that the \ac{CRB} of the TDoAs, distance, and angle can be obtained through parameter transformation from \eqref{eq:J_delay}.

\subsection{Closed-form CRB of TDoA}\label{subsec:CRB_tdoa}
\begin{proposition}
The \ac{CRB} for the TDoA is given by
\begin{equation}
	\begin{aligned}
		C(\delta_m) 
		&= \delta_m^{(\tilde{\boldsymbol{\theta}}_\tau)} 
		\mathbf{J}(\tilde{\boldsymbol{\theta}}_\tau)^{-1} 
		(\delta_m^{(\tilde{\boldsymbol{\theta}}_\tau)})^{\rm T} 
		= \frac{\sigma^2}{\alpha} \\[3pt]
		&= 
		\frac{\sigma^2}
		{|\gamma|^2 (2\pi)^2 N} 
		\cdot 
		\frac{12}{f_0^2 (N^2-1) + 12 f_{\rm c}^2},
	\end{aligned}
	\label{eq:C_delta}
\end{equation}
where $\delta_m^{(\tilde{\boldsymbol{\theta}}_\tau)} = \frac{\partial \delta_m}{\partial \tilde{\boldsymbol{\theta}}_\tau^{\rm T}} 
= \mathbf{e}_m^{\rm T} - \mathbf{e}_0^{\rm T} $, and $\mathbf{e}_m = [\,0,\, \ldots,\, 0,\, 1,\, 0,\, \ldots,\, 0\,]^{\rm T} \in \mathbb{R}^{(M+2)}$ denotes the $m$-th canonical basis (or unit) vector of dimension $M+2$, i.e., a vector whose $n$th element is $\{\mathbf{e}_m\}_n = 1$ if $n=m$ and zero otherwise.
%
%\[
%\{\mathbf{e}_m\}_n = 
%\begin{cases}
%	1, & n = m, \\[3pt]
%	0, & n \neq m.
%\end{cases}
%\]
%
%For completeness, recall that $\delta_m = \tau_m - \tau_0$, which implies that 
%$\delta_m^{(\tilde{\boldsymbol{\theta}}_\tau)} 
%= [-1,\, 0,\, \cdots,\, 1,\, \cdots]$ 
%has $-1$ in the first position (corresponding to $\tau_0$) and $1$ in the $m$th position.
\end{proposition}
\begin{proof}
	Using \eqref{eq:J_delay}, the transformation in \eqref{eq:C_delta} can, without loss of generality, be written as
	\begin{equation}
		\begin{aligned}
			C(\delta_m) 
			&= 
			\begin{bmatrix}
				-1 & 1
			\end{bmatrix}
			\!\left(
			\frac{1}{\alpha}
			\begin{bmatrix}
				1 & 0 \\[2pt] 0 & 1
			\end{bmatrix}
			+
			\frac{\beta}{\alpha^2}
			\begin{bmatrix}
				1 & 1 \\[2pt] 1 & 1
			\end{bmatrix}
			\right)
			\!
			\begin{bmatrix}
				-1 \\[2pt] 1
			\end{bmatrix}
%			\\[3pt]
%			&= \frac{\sigma^2}{2}\cdot\frac{2}{\alpha}, 
%			\qquad \forall\, m.
		\end{aligned}
		\label{eq:prop2_proof}
	\end{equation}
	 $\forall\, m$. The term associated with $\beta / \alpha^2$ cancels because the left and right vectors have opposite signs, while the $1/\alpha$ term remains and is scaled by $2$, since
	$[-1 \,\, 1] [-1 \,\, 1]^{\rm T} = 2$.
\end{proof}

The structure of the TDoA CRB in \eqref{eq:C_delta} is very insightful. 
The parameter $\alpha$ in \eqref{eq:alpha} captures how signal parameters affect the CRB. 
Moreover, it shows how enforcing the carrier-frequency constraint through $b(\tau)$ in \eqref{eq:ab} significantly reduces the CRB compared to the case where the constraint is ignored (e.g., setting $f_{\rm c}=0$). Nonetheless, the simplified model with $f_{\rm c}=0$ remains useful for the CRB achievability analysis in Section~\ref{sec:estimator_R_phi}.

\subsection{CRBs of $\phi$ and $R$}\label{subsec:CRBs_phi_R}
%
%In this subsection, we first derive a intermediate matrix-form \ac{CRB} expression for $\phi$ and $R$, which is further computed in closed-form for ULA.
%First, we write the Jacobian of ${\tilde{\boldsymbol{\upsilon}}} = [\,\tau_0 \,\, \phi \,\, R \,\, \tilde{\gamma}^{\rm r} \,\, \tilde{\gamma}^{\rm i}\,]^{\rm T}$ as
In this subsection, we first derive an intermediate matrix-form \ac{CRB} expression for $\phi$ and $R$, which is later specialized to the ULA case. The Jacobian of 
$\tilde{\boldsymbol{\upsilon}} = [\,\tau_0 \,\, \phi \,\, R \,\, \tilde{\gamma}^{\rm r} \,\, \tilde{\gamma}^{\rm i}\,]^{\rm T}$
is written as
\begin{equation}
	\mathbf{D}(\tilde{\boldsymbol{\upsilon}}) 
	= \mathbf{D}(\tilde{\boldsymbol{\theta}}_\tau(\boldsymbol{\upsilon}))\, \tilde{\boldsymbol{\theta}}_\tau^{(\tilde{\boldsymbol{\upsilon}})},
	\label{eq:D_upsilon_app}
\end{equation}
%
%where $\tilde{\boldsymbol{\theta}}_\tau^{(\tilde{\boldsymbol{\upsilon}})} \in \mathbb{R}^{(M+2) \times 5}$ is given by
where $\tilde{\boldsymbol{\theta}}_\tau^{(\tilde{\boldsymbol{\upsilon}})} \in \mathbb{R}^{(M+2)\times 5}$ is 
\begin{equation}
	\tilde{\boldsymbol{\theta}}_\tau^{(\tilde{\boldsymbol{\upsilon}})} 
	= \frac{\partial \tilde{\boldsymbol{\theta}}_\tau(\boldsymbol{\upsilon})}{\partial \tilde{\boldsymbol{\upsilon}}^{\rm T}}
	= 
	\begin{bmatrix}
		\mathbf{U} & \mathbf{0}_{M,2} \\[2pt]
		\mathbf{0}_{3,2} & \mathbf{I}_2
	\end{bmatrix},
	\label{eq:theta_upsilon}
\end{equation}
\begin{equation}
	\text{and }\mathbf{U} 
	= [\,\mathbf{1}_M \,\, \boldsymbol{\delta}_0^{(\phi)} \,\, \boldsymbol{\delta}_0^{(R)}\,] 
	\in \mathbb{R}^{M \times 3}
\end{equation}
%
%is the sub-matrix that converts the delay vector to the parameters $[\,\tau_0 \,\, \phi \,\, R\,]^{\rm T}$. 
%This follows from the fact that 
%$\tilde{\boldsymbol{\theta}}_\tau 
%= [\,(\delta_0 + \tau_0) \,\, (\delta_1 + \tau_0) \,\, \cdots \,\, (\delta_{M-1} + \tau_0) \,\, \tilde{\gamma}^{\rm r} \,\, \tilde{\gamma}^{\rm i}\,]^{\rm T}$.
%Lastly, the angle and Distance-related partial derivative vectors of size $M$, with zero in their first element, are  $\boldsymbol{\delta}_0^{(\phi)} = [0\,\, (\boldsymbol{\delta}^{(\phi)})^{\rm T}]^{\rm T}$ and $\boldsymbol{\delta}_0^{(R)} = [0\,\, (\boldsymbol{\delta}^{(R)})^{\rm T}]^{\rm T}$.
%This reflects the fact an equivalent way of writing the TDoA model is to set $\delta_0 = 0$.
maps the delay vector to the parameters $[\,\tau_0 \,\, \phi \,\, R\,]^{\rm T}$.  
This follows from 
$\tilde{\boldsymbol{\theta}}_\tau 
= [\,(\delta_0+\tau_0)\,\,\cdots\,\,(\delta_{M-1}+\tau_0)\,\,\tilde{\gamma}^{\rm r}\,\,\tilde{\gamma}^{\rm i}\,]^{\rm T}$.
The derivative vectors $\boldsymbol{\delta}_0^{(\phi)} = [0\,\,(\boldsymbol{\delta}^{(\phi)})^{\rm T}]^{\rm T}$ and $\boldsymbol{\delta}_0^{(R)} = [0\,\,(\boldsymbol{\delta}^{(R)})^{\rm T}]^{\rm T}$ have zeros in their first entry, reflecting the equivalent representation $\delta_0 = 0$.

Combining \eqref{eq:J_delay}, \eqref{eq:D_upsilon_app}--\eqref{eq:theta_upsilon}, \eqref{eq:J_theta_delay_app}, the identities in Appendix~\ref{app:inv_FIM_theta_delay}, and \eqref{eq:D_upsilon_app}--\eqref{eq:theta_upsilon}, the inverse FIM takes the form

\begin{equation}
	\mathbf{J}(\tilde{\boldsymbol{\upsilon}})^{-1} 
	= \frac{\sigma^2}{2} 
	\begin{bmatrix}
		\frac{1}{\alpha}(\mathbf{U}^{\rm T}\mathbf{U})^{-1} 
		+ \frac{\beta}{\alpha^2}\mathbf{W} 
		& \cdots \\[3pt]
		\cdots & \ddots
	\end{bmatrix},
	\label{eq:J_theta_upsilon}
\end{equation}
where $
\mathbf{W} 
= (\mathbf{U}^{\rm T}\mathbf{U})^{-1} 
\mathbf{U}^{\rm T}\mathbf{1}_M 
\mathbf{1}_M^{\rm T}\mathbf{U} 
(\mathbf{U}^{\rm T}\mathbf{U})^{-1}$. 
This structure enables a direct extraction of the CRBs for $\phi$ and $R$.
\begin{proposition}
The \ac{CRB} values of $\phi$ and $R$ are 
\begin{equation}
	\begin{aligned}
		C(\phi) 
		&= [\,\mathbf{J}(\upsilon)^{-1}\,]_{2,2} 
		= \frac{\sigma^2}{2\alpha} 
		[(\mathbf{U}^{\rm T}\mathbf{U})^{-1}]_{2,2} \\[2pt]
		\text{and } C(R) 
		&= [\,\mathbf{J}(\upsilon)^{-1}\,]_{3,3} 
		= \frac{\sigma^2}{2\alpha} 
		[(\mathbf{U}^{\rm T}\mathbf{U})^{-1}]_{3,3}.
	\end{aligned}
	\label{eq:CRB_phi_R}
\end{equation}
\end{proposition}
\begin{proof}
	Let $\mathbf{e}_0 = [\,1 \,\, 0 \,\, 0\,]^{\rm T}$ denote the first canonical (unit) vector in $\mathbb{R}^3$. 
	Because the first column of $\mathbf{U}$ is $\mathbf{1}_M$, we can write $\mathbf{U}\mathbf{e}_0 = \mathbf{1}_M$. 
	Substituting this into the definition of $\mathbf{W}$ gives 
	$\mathbf{W} = \mathbf{e}_0 \mathbf{e}_0^{\rm T}$, 
	such that $[\mathbf{W}]_{2,2} = [\mathbf{W}]_{3,3} = 0$.
\end{proof}
The CRBs in \eqref{eq:CRB_phi_R} cleanly separate signal-dependent terms (sub-carrier spacing, center frequency, and number of sub-carriers) through $\alpha$ from array-dependent terms (geometry, spacing, and size) through $\mathbf{U}$.  
This modularity enables separate optimization of waveform and array design.  
The expression also supports arbitrary array geometries through $\mathbf{U}$, which is useful for unconventional designs such as movable or irregular antenna layouts.

Next, we derive approximate closed-form expressions for the ULA, explicitly showing the dependence on array parameters.
	
\begin{proposition}
The approximate \ac{CRB} expressions for the ULA are
\begin{equation}
	\begin{aligned}
		C_{\rm ULA}(R) 
		 = \frac{\sigma^2}{2\alpha} 
		\left( \frac{2 R^2 c}{\delta_{\rm x}^2 \sin^2 \phi} \right)^2 \!\!
		\frac{180}{M (M^2 - 1)(M^2 - 4)},
	\end{aligned}
	\label{eq:C_R_ULA}
\end{equation}
\noindent The units are consistent since $\alpha$ scales as $|\tilde{\gamma}|^2\,\mathrm{Hz}^2$, while the geometric terms originate from $\partial \delta_m/\partial R$ which includes $1/c$.
\begin{equation}
	\begin{aligned}
		\text{and } C_{\rm ULA}(\phi) 
		= \frac{\sigma^2}{2\alpha} 
		\left( \frac{c}{\delta_{\rm x} \sin \phi} \right)^2 
		\frac{12}{M (M^2 - 1)}.
	\end{aligned}
	\label{eq:C_phi_ULA}
\end{equation}
\end{proposition}
\begin{proof}
	See Appendix~\ref{app:CRB_ULA} for details.
\end{proof}
%
%These bounds naturally reflect the transition between near- and far-field regimes.  
The CRB for the distance shows that the RMSE scales quadratically with $R$: doubling $R$ increases the RMSE by a factor of four for the same \ac{SNR}\footnote{In practice, SNR decreases with distance, so the degradation is larger. Path loss is omitted here to focus on the estimator structure.}.  
This behavior is expected, as sufficiently large $R$ places the transmitter in the far field, where distance becomes unidentifiable.
The antenna spacing has an opposite effect: doubling $\delta_{\rm x}$ reduces the RMSE by a factor of four.  
%Intestingly, this holds for any array geometry because scaling the antenna coordinates in \eqref{eq:partial_delta_m} is equivalent to scaling $R$.
The dependence on $1/\sin^2\phi$ indicates degraded performance at large $\phi$ for ULAs.  
This limitation can be mitigated by two-dimensional arrays spanning both the $x$- and $y$-axes, which is left to be generalized in future work.

The AoA CRB is independent of $R$ and also scales with $1/\sin^2\phi$.  
Its form closely matches the classical far-field result \cite{RimaxThesis}, with only a marginal improvement due to the $f_0$ term in $\alpha$, which is negligible because $f_{\rm c}$ dominates.  
A key advantage of the per-antenna model is its ability to resolve angular ambiguities caused by far-field grating lobes, which we leave for future investigation.

\section{Estimator and CRB Achievability}\label{sec:estimator_R_phi}
In this section, we use the expressions developed in Section~\ref{sec:system_model_single_band} to derive a benchmark estimator. 
Subsection~\ref{subsec:general_ML} introduces a general estimation framework under AWGN based on the iterative \ac{LM} method, which serves as the core of the proposed estimator. 
Subsection~\ref{subsec:3stage_algorithm} develops the three-stage estimator for the single-band AoA–distance model. 
Finally, Subsection~\ref{subsec:TSNRs} provides approximate SNR thresholds required for the estimator to achieve the TDoA and distance CRBs, respectively.
\subsection{General Maximum Likelihood Estimator under AWGN}\label{subsec:general_ML}
Consider the generic complex AWGN model
\begin{equation}
\mathbf{y}_\kappa = \mathbf{s}_\kappa(\boldsymbol{\kappa}) + \mathbf{w}, 
\qquad \mathbf{w}\sim\mathcal{CN}(\mathbf{0},\sigma^2\mathbf{I}_{N_\kappa}),
\label{eq:awgn_generic}
\end{equation}
where $\mathbf{y}_\kappa\in\mathbb{C}^{N_\kappa}$ and $\mathbf{s}_\kappa(\boldsymbol{\kappa})\in\mathbb{C}^{N_\kappa}$. The corresponding likelihood is
\begin{equation}
p(\mathbf{y}_\kappa \mid \boldsymbol{\kappa})
= \frac{1}{(\pi\sigma^2)^{N_\kappa}}
\exp\!\left(-\frac{\|\mathbf{y}_\kappa-\mathbf{s}_\kappa(\boldsymbol{\kappa})\|_2^2}{\sigma^2}\right).
\label{eq:likelihood_generic}
\end{equation}
where $\boldsymbol{\kappa}$ is a generic parameter vector, and $\mathbf{s}_{{\kappa}}(\cdot)$ determines the system model. The Jacobian matrix is $\mathbf{D}_\kappa(\boldsymbol{\kappa})=\partial \mathbf{s}_\kappa(\boldsymbol{\kappa})/\partial \boldsymbol{\kappa}^{\rm T}$. 
The maximum likelihood (ML) estimation of $\boldsymbol{\kappa}$ is $\hat{\boldsymbol{\kappa}} = \operatorname*{arg\,max}\limits_{\boldsymbol{\kappa}} p(\mathbf{y}_\kappa \mid \boldsymbol{\kappa})$.
%Several numerical methods have been proposed in the literature to solve problem in the form of~\eqref{eq:arg_max_kappa}. 
Assuming the CRB regularity conditions are satisfied, i.e.,  continuity of the model and compact parameter space, a suitable benchmark method to estimate $\hat{\boldsymbol{\kappa}}$ is the iterative LM algorithm described in~\cite[Table 5-3]{RimaxThesis}, which achieves the \ac{CRB} for multiple cases, e.g., the delay, \ac{AoD} and \ac{AoA} \cite{Bomfin_DMC_globecom}.

%The LM algorithm refines an initial estimate of the parameter vector by iteratively minimizing the error between the model and the observations.
%When not available, the initial estimates are typically obtained via search on a discrete subset of the parameter space.
%The LM method takes as input the observed data, $\mathbf{y}_{{\kappa}}$, an initial parameter estimate, $\hat{\boldsymbol{\kappa}}_{\rm ini}$, the model that relates the parameters to the observations, $\mathbf{s}_{{\kappa}}(\cdot)$, and the corresponding Jacobian of this model with respect to the parameters, $\mathbf{D}_{{\kappa}}(\cdot)$.

When an initial estimate $\hat{\boldsymbol{\kappa}}_{\rm ini}$ is available, the input-output relation of the LM algorithm can be written in general as
\begin{equation}
	\hat{\boldsymbol{\kappa}}_{\rm LM} = {\rm LM}(\mathbf{y}_{{\kappa}},\hat{\boldsymbol{\kappa}}^{\rm ini};\mathbf{s}_{{\kappa}}(\cdot),\mathbf{D}_{{\kappa}}(\cdot)).
\end{equation}
A detailed definition of the LM algorithm is described in~\cite[Table 5-3]{RimaxThesis}.
For brevity, we use the compact notation above for a general parameter $\boldsymbol{\kappa}$, which makes explicit the inputs $(\mathbf{y}_{{\kappa}},\hat{\boldsymbol{\kappa}}^{\rm ini})$, as well as the underlying model $(\mathbf{s}_{{\kappa}}(\cdot),\mathbf{D}_{{\kappa}}(\cdot))$. 

\begin{algorithm}[t!]
	\small 
	\caption{Single-Band 3-Stages LM Algorithm}
	\label{alg:single_band}
	\begin{algorithmic}[1]  % The number tells where line numbering starts
		\Input Observation $\mathbf{y} = [\mathbf{y}_0^{\rm T} \,\, \mathbf{y}_1^{\rm T}  \cdots \mathbf{y}_{M-1}^{\rm T}]^{\rm T}$	
		\Statex $/\star$ \texttt{\textcolor{brown}{Stage-1: Delay, \eqref{eq:s_tau}}} $\hfill \star/$
		\For{$m = 0$ to $M-1$}
		 	\State Initialize ${\boldsymbol{\theta}}_{\tau_m}^{\rm {\rm ini}} = [{\tau}_{m}^{\rm {\rm ini}} \,\, {\gamma}_{m}^{\rm r, ini} \,\, {\gamma}_{m}^{\rm i, ini}]^{\rm T}$ via FFT-based fix grid search \cite{RimaxThesis} taking $\mathbf{y}_m$ as input
			\State Refine $\hat{\boldsymbol{\theta}}_{\tau_m,{\rm LM}} = {\rm LM}(\mathbf{y}_m,\boldsymbol{\theta}_{\tau_m}^{\rm ini};\mathbf{s}_{{\tau}}(\cdot),\mathbf{D}_{{\tau}}(\cdot))$
		\EndFor	
		\Statex $/\star$ \texttt{\textcolor{brown}{Stage-2: TDoA, \eqref{eq:s_theta} and \eqref{eq:D_theta}}} $\hfill \star/$
		\State Initialize ${\boldsymbol{\theta}}^{\rm ini}$ using ${\delta}_{m}^{\rm ini} = \hat{\tau}_{m,{\rm LM}} - \hat{\tau}_{0,{\rm LM}}$ 
		\State Refine $\hat{\boldsymbol{\theta}}_{\rm LM} = {\rm LM}(\mathbf{y},{\boldsymbol{\theta}}^{\rm ini};\mathbf{s}_{\theta}(\cdot),\mathbf{D}_{\theta}(\cdot))$
		\Statex  $/\star$ \texttt{\textcolor{brown}{Stage-3: Dist./Ang., \eqref{eq:s_theta} and \eqref{eq:D_upsilon2}}} $\hfill \star/$
		\State Initialize ${\boldsymbol{\upsilon}}^{\rm ini}$ using $\hat{\boldsymbol{\theta}}_{\rm LM}$ (Eq. \eqref{eq:R_phi_hat_ini})
		\State Refine $\hat{\boldsymbol{\upsilon}}_{\rm LM} = {\rm LM}(\mathbf{y},{\boldsymbol{\upsilon}}^{\rm ini};\mathbf{s}_{\theta(\upsilon)}(\cdot),\mathbf{D}_{\theta(\upsilon)}(\cdot))$
		\Output $\hat{\boldsymbol{\upsilon}}_{\rm LM}$
	\end{algorithmic}
\end{algorithm}

\subsection{3-Stage Algorithm to Estimate $R$ and $\phi$}\label{subsec:3stage_algorithm}
The goal of this subsection is to derive an estimator for $R$ and $\phi$, or equivalently for $\boldsymbol{\upsilon}$. 
Using the LM algorithm to estimate $\boldsymbol{\upsilon}$ requires an initial value $\hat{\boldsymbol{\upsilon}}^{\rm ini}$, which is not straightforward to obtain. 
Since estimating $\boldsymbol{\theta}$ is simpler, a natural approach is to first estimate the TDoA vector $\boldsymbol{\theta}$ and then invert \eqref{eq:delta_R_pi}. 
However, refining $\boldsymbol{\theta}$ with LM also requires an initial estimate $\hat{\boldsymbol{\theta}}^{\rm ini}$, which can be obtained by estimating the delays per receive antenna.

Motivated by this dependency, we propose the 3-stage algorithm in Algorithm~\ref{alg:single_band}. 
\textit{Stage-1} estimates the per-antenna delays, \textit{Stage-2} estimates the TDoA vector using the delay estimates for initialization, and \textit{Stage-3} estimates $R$ and $\phi$ using the refined TDoAs.

\subsubsection{Stage-1: Delay}
The first stage estimates the delay at each receive antenna. 
Define the per-antenna parameter vector 
${\boldsymbol{\theta}}_{\tau_m} = [{\tau}_{m} \,\, \gamma_{m}^{\rm r} \,\, \gamma_{m}^{\rm i}]^{\rm T}$, 
with model and Jacobian
\begin{equation}
\begin{aligned}
\mathbf{s}_{\tau}({\boldsymbol{\theta}}_{\tau_m}) 
&= \mathbf{a}(\tau_m)\,\gamma_m,\\
\mathbf{D}_{\tau}({\boldsymbol{\theta}}_{\tau_m}) 
%&= \left[\frac{\partial \mathbf{s}_\tau}{\partial \tau_m}\;\; \frac{\partial \mathbf{s}_\tau}{\partial \gamma_m^{\rm r}}\;\; \frac{\partial \mathbf{s}_\tau}{\partial \gamma_m^{\rm i}}\right]
& = \left[\boldsymbol{\Xi}(f_0,0)\mathbf{a}(\tau_m)\gamma_m \;\; \mathbf{a}(\tau_m)\;\; j\mathbf{a}(\tau_m)\right],
\end{aligned}
\label{eq:s_tau}
\end{equation}
obtained from \eqref{eq:s_m} using $\gamma_m = b(\tau_m)\gamma$. 
This Stage-1 model intentionally neglects the $f_{\rm c}$-dependent phase term, since it is highly sensitive to coarse delay errors and is instead exploited in Stage-2.
The initial estimate $\hat{\boldsymbol{\theta}}_{\tau_m}^{\rm ini}$ is obtained via a fixed-grid delay search, efficiently implemented through an oversampled IFFT \cite{RimaxThesis,Bomfin_UW}. 
The LM refinement is then computed as 
$\hat{\boldsymbol{\theta}}_{\tau_m,{\rm LM}} = {\rm LM}(\mathbf{y}_m,\hat{\boldsymbol{\theta}}_{\tau_m}^{\rm ini};\mathbf{s}_{\tau}(\cdot),\mathbf{D}_{\tau}(\cdot))$.

\subsubsection{Stage-2: TDoA}
Assuming sufficiently accurate delay estimates, we initialize
\begin{equation}
\hat{\boldsymbol{\theta}}^{\rm ini}
=\left[\hat{\tau}_0^{\rm ini}\; \hat{\delta}_1^{\rm ini}\; \cdots\; \hat{\delta}_{M-1}^{\rm ini}\; \hat{\gamma}^{\rm r,ini}\; \hat{\gamma}^{\rm i,ini}\right]^{\rm T},
\hat{\delta}_m^{\rm ini}=\hat{\tau}_{m,{\rm LM}}-\hat{\tau}_{0,{\rm LM}}.
\end{equation}

The channel coefficient is initialized as  
$\hat{\gamma}^{\rm ini} = \hat{\gamma}^{\rm r,ini} + j \hat{\gamma}^{\rm i,ini}$  
by correlating the unit-gain model ($\gamma=1$) with the received signal.  
The LM refinement is given by  
$\hat{\boldsymbol{\theta}}_{\rm LM} = {\rm LM}(\mathbf{y},\hat{\boldsymbol{\theta}}^{\rm ini};\mathbf{s}_{\theta}(\cdot),\mathbf{D}_{\theta}(\cdot))$,  
using the model and Jacobian in \eqref{eq:s_theta} and \eqref{eq:D_theta}.
\subsubsection{Stage-3: Distance and AoA}
The initial estimates $\hat{R}^{\rm ini}$ and $\hat{\phi}^{\rm ini}$ are obtained by inverting \eqref{eq:delta_R_pi}. 
Since $R$ and $\phi$ are two unknowns, at least two TDoAs are needed. The most informative ones correspond to the antennas farthest from the reference antenna at $x_0=0$, as they provide the largest effective aperture.
Let $m_1$ and $m_2$ denote these antennas.  
Solving \eqref{eq:delta_R_pi} for $m_1$ and $m_2$ yields
\begin{equation}
	\begin{aligned}
		\hat{R}^{\rm ini} &=
		\frac{-c^2\hat{\delta}_{m_1}^2 x_{m_2} 
			+ c^2\hat{\delta}_{m_2}^2 x_{m_1}  
			+ x_{m_1}^2x_{m_2} - x_{m_1}x_{m_2}^2}
		{2(c\hat{\delta}_{m_1}x_{m_2} - c\hat{\delta}_{m_2}x_{m_1})}, \\[2pt]
		\hat{\phi}^{\rm ini} &=
		\cos^{-1}\!\left(\frac{c^3\hat{\delta}_{m_1}^2\hat{\delta}_{m_2}
			- c^3\hat{\delta}_{m_1}\hat{\delta}_{m_2}^{2}
			+ c\hat{\delta}_{m_1}x_{m_2}^{2}
			- c\hat{\delta}_{m_2}x_{m_1}^{2}}
		{-c^2\hat{\delta}_{m_1}^2 x_{m_2}
			+ c^2\hat{\delta}_{m_2}^2 x_{m_1}
			+ x_{m_1}^2x_{m_2}
			- x_{m_1}x_{m_2}^2}\right),
	\end{aligned}
	\label{eq:R_phi_hat_ini}
\end{equation}
computed using $\hat{\delta}_{m_1,{\rm LM}}$ and $\hat{\delta}_{m_2,{\rm LM}}$ from Stage-2 to form $\hat{\boldsymbol{\upsilon}}^{\rm ini}$.
Lastly, the final estimates are refined as  
$\hat{\boldsymbol{\upsilon}}_{\rm LM} = {\rm LM}(\mathbf{y},\hat{\boldsymbol{\upsilon}}^{\rm ini};\mathbf{s}_{\theta(\upsilon)}(\cdot),\mathbf{D}_{\theta(\upsilon)}(\cdot))$,  
where \eqref{eq:s_theta} uses $\boldsymbol{\theta}(\boldsymbol{\upsilon})$ and the Jacobian is given in \eqref{eq:D_upsilon2}.

\subsection{CRB Achievability via TSNR}\label{subsec:TSNRs}
ML estimators typically exhibit poor performance at low SNR, followed by a sharp MSE drop once the SNR exceeds a certain threshold, after which the CRB is closely approached. This threshold is known as the \ac{TSNR}.  
In this section, we characterize approximate TSNR values at which the three-stage estimator in Algorithm~\ref{alg:single_band} operates near the CRB. 
The analysis focuses on the TDoA and distance components, while AoA TSNR analysis is left for future work due to ambiguity-related challenges.

\subsubsection{TDoA TSNR in Stage-2}\label{subsubsec:TSNR_tdoa}
Recall that Stage-1 estimates per-antenna delays using the model in \eqref{eq:s_tau}, which does not exploit the $f_{\rm c}$-dependent phase.  
To assess TDoA achievability in Stage-2, we ask: \emph{what SNR is required for the estimates from Stage-1 delay to be accurate enough for the LM refinement in Stage-2?}  
The following proposition provides an approximate threshold.
\begin{proposition}
	Let $\rho = |\gamma|^2/\sigma^2$ denote the SNR.  
	The SNR required for TDoA estimation in Stage-2 of Algorithm~\ref{alg:single_band} to achieve the CRB in \eqref{eq:C_delta} can be approximated by
	\begin{equation}
		\rho> \rho_{\delta}^\star \approx \frac{12 f^2}{N f_0^2 (N^2-1)} \approx \frac{ 12}{N} \frac{f_{\rm c}^2}{B^2},
		\label{eq:rho_delta_star}
	\end{equation}
	where $B = N f_0$ is the signal bandwidth. 
\end{proposition}

\noindent\textbf{Rationale:} See Appendix~\ref{app:tdoa_waterfall}. \hfill$\square$
%The derivation in Appendix~\ref{app:tdoa_waterfall} is heuristic rather than rigorous.  

The expression in \eqref{eq:rho_delta_star} provides useful insights.
For example, the TSNR scales with $f_{\rm c}^2$, with the inverse square of the bandwidth, and with $1/N$.  
Thus, higher carrier frequencies require higher SNR for accurate TDoA estimation.  
This suggests a natural strategy for multi-band localization: use lower-frequency sub-bands to reduce the TSNR and higher-frequency sub-bands to reduce the MSE via \eqref{eq:C_delta} and \eqref{eq:alpha}.  
For a fixed bandwidth, increasing $N$ reduces sub-carrier spacing and increases the processing gain, which explains the $1/N$ scaling.

Lastly, the results of Section~\ref{sec:results} provide an analysis that supports the derivation and offer further intuition.

\subsubsection{Distance TSNR}\label{subsubsec:TSNR_distance}
Next, we ask: \emph{what SNR is required for the TDoA estimates $\hat{\delta}_{m_1,{\rm LM}}$ and $\hat{\delta}_{m_2,{\rm LM}}$ estimated by Stage-2 to yield an initialization $\hat{R}^{\rm ini}$ in \eqref{eq:R_phi_hat_ini} accurate enough for LM refinement in Stage-3?}  
The answer is given below.
\begin{proposition}
	The SNR required for the Stage-2 TDoA estimates to enable successful distance refinement in Stage-3 is approximately
	\begin{equation}
		\rho> \rho_{\rm R}^\star	\approx \frac{1}{N}\max\left(\frac{(4 R c)^2 }{f_{\rm c}^2 (M-1)^4 \delta_{\rm x}^4  \sin^4 \phi},\frac{12f_{\rm c}^2}{B^2}\right).
		\label{eq:rho_R_star}
	\end{equation}
\end{proposition}
\noindent\textbf{Rationale:} See Appendix~\ref{app:dist_waterfall}. \hfill$\square$

%Similarly to the TDoA TSNR, the result in \eqref{eq:rho_R_star} is derived heuristically.  
A key point is that $\hat{R}^{\rm ini}$ in \eqref{eq:R_phi_hat_ini} relies on a three-antenna sub-array ($M=3$), and thus the achievable TSNR corresponds to this sub-array configuration.  
Despite this nuance, the distance TSNR remains a useful reference for distinguishing the SNR regime in which the estimator exhibits large errors from the regime in which it closely approaches the CRB for any $M \ge 3$.

Lastly, the results of Section~\ref{sec:results} provide intermediate numerical results that support the derivation and offer further intuition.

\section{Multi-band Generalization} \label{sec:system_model_multi_band}
%
%In the previous sections, we focused on the single-band model of Section~\ref{sec:system_model_single_band}.  
In this section, we generalize the single-band results to the multi-band case. 
Subsection~\ref{subsec:multi_band_model} presents the generalized multi-band model, followed by the matrix-form FIMs for TDoA and AoA-distance in Subsection~\ref{subsec:multi_band_matrix_form}.  
Closed-form CRB generalizations are derived in Subsection~\ref{subsec:multi_band_closed_form_CRBs}, and the multi-band estimator and TSNR extensions are discussed in Subsection~\ref{subsec:multi_band_estimator}.
\subsection{General Multi-band Model}\label{subsec:multi_band_model}
Let $Q$ denote the number of available sub-bands.  
For each band $q$, let $f_{{\rm c}_q}$, $f_{0_q}$, and $N_q$ represent the center frequency, subcarrier spacing, and number of subcarriers, respectively.  
All bands experience the same delays, but each band has its own channel gain $\gamma_q$.  
The overall multi-band parameter vector is formed by concatenation as%
\begin{equation}
	\boldsymbol{\theta}_q 
	= [\, \tau_0 \,\, \delta_1 \,\, \cdots \,\, \delta_{M-1} \,\, \gamma_q^{\rm r} \,\, \gamma_q^{\rm i} \,]^{\rm T}.
	\label{eq:theta_q}
\end{equation}
The single-band model in \eqref{eq:s_theta} extends naturally to the multiband case as
\begin{equation}
	\mathbf{s}(\boldsymbol{\Theta})
	\!= \!\big[\, 
	\mathbf{s}_0(\boldsymbol{\theta}_0)^{\rm T} \,\,
	\mathbf{s}_1(\boldsymbol{\theta}_1)^{\rm T} \,
	\cdots \,
	\mathbf{s}_{Q-1}(\boldsymbol{\theta}_{Q-1})^{\rm T}
	\big]^{\rm T}\!\!
	\in \!\mathbb{C}^{M\cdot \sum_q N_q},
	\label{eq:s_Theta}
\end{equation}
which concatenates the single-band received vectors into a larger multi-band observation.
The overall multi-band parameter vector $\boldsymbol{\Theta} = \bigcup_q \boldsymbol{\theta}_q$ can be written as 
%$\boldsymbol{\Theta}
%		= [\, \dot{\boldsymbol{\theta}}^{\rm T} \,\, \boldsymbol{\Gamma}^{\rm T} \,]^{\rm T}$, where $\dot{\boldsymbol{\theta}}
%		= [\, \tau_0 \,\, \delta_1 \,\, \cdots \,\, \delta_{M-1} \,]^{\rm T}$ and $\boldsymbol{\Gamma}
	%	 = [\, \gamma_0^{\rm r} \,\, \gamma_0^{\rm i} \, \cdots \, \gamma_{Q-1}^{\rm r} \,\, \gamma_{Q-1}^{\rm i} \,]^{\rm T}$.
%
\begin{equation}
	\begin{aligned}
		\boldsymbol{\Theta}
		&= [\, \dot{\boldsymbol{\theta}}^{\rm T} \,\, \boldsymbol{\Gamma}^{\rm T} \,]^{\rm T}, \\[3pt]
		\dot{\boldsymbol{\theta}}
		&= [\, \tau_0 \,\, \delta_1 \,\, \cdots \,\, \delta_{M-1} \,]^{\rm T}, \\[2pt]
		\text{and }\boldsymbol{\Gamma}
		& = [\, \gamma_0^{\rm r} \,\, \gamma_0^{\rm i} \, \cdots \, \gamma_{Q-1}^{\rm r} \,\, \gamma_{Q-1}^{\rm i} \,]^{\rm T}.
	\end{aligned}
	\label{eq:Theta_q}
\end{equation}

\subsection{Multi-band Matrix-form Jacobian and FIM}\label{subsec:multi_band_matrix_form}
This subsection generalizes the expressions of Section~\ref{sec:crb_R_phi} (single-band) to the multi-band model.  
The Jacobians with respect to the subvectors $\dot{\boldsymbol{\theta}}$ and $\boldsymbol{\Gamma}$ are
\begin{equation}
	\begin{aligned}
		\mathbf{D}(\dot{\boldsymbol{\theta}}) 
		&= 
		\big[
		\mathbf{D}_{0}(\dot{\boldsymbol{\theta}})^{\rm T} \,\,
		\mathbf{D}_{1}(\dot{\boldsymbol{\theta}})^{\rm T} \,\,
		\cdots \,\,
		\mathbf{D}_{Q-1}(\dot{\boldsymbol{\theta}})^{\rm T}
		\big]^{\rm T}, \\[3pt]
		\text{and }\mathbf{D}(\boldsymbol{\Gamma}) 
		&=
		\begin{bmatrix}
			\mathbf{D}_{0}(\gamma_{0}) & \mathbf{0} & \cdots \\
			\mathbf{0} & \mathbf{D}_{1}(\gamma_{1}) & \cdots \\
			\vdots & \ddots & \cdots \\
			\mathbf{0} & \cdots & \mathbf{D}_{Q-1}(\gamma_{Q-1})
		\end{bmatrix},
	\end{aligned}
	\label{eq:D_multiband}
\end{equation}
where  $\mathbf{D}_q(\dot{\boldsymbol{\theta}}) = [\,\mathbf{S}_q^{(\tau_0)} \;\; \mathbf{S}_q^{(\boldsymbol{\delta})}\,]$ and $\mathbf{D}_q(\gamma_q) = \mathbf{S}_q^{(\gamma_q)}$. The partial derivatives $\mathbf{S}_q^{(\tau_0)}$, $\mathbf{S}_q^{(\boldsymbol{\delta})}$ and $\mathbf{S}_q^{(\gamma_q)}$ have the same definition as in \eqref{eq:S_delta_derivatives}, but for the $q$th sub-band.
The multi-band Jacobian and corresponding FIM are
\begin{equation}
	\begin{aligned}
		\mathbf{D}(\boldsymbol{\Theta}) 
		&= [\,\mathbf{D}(\dot{\boldsymbol{\theta}}) \;\; \mathbf{D}(\boldsymbol{\Gamma})\,], \\[3pt]
		\text{and } \mathbf{J}(\boldsymbol{\Theta}) 
		&= \frac{2}{\sigma^{2}} 
		\Re\!\left(
		\begin{bmatrix}
			\mathbf{D}(\dot{\boldsymbol{\theta}})^{\rm H}\mathbf{D}(\dot{\boldsymbol{\theta}}) 
			& \mathbf{D}(\dot{\boldsymbol{\theta}})^{\rm H}\mathbf{D}(\boldsymbol{\Gamma}) \\
			\mathbf{D}(\boldsymbol{\Gamma})^{\rm H}\mathbf{D}(\dot{\boldsymbol{\theta}})
			& \mathbf{D}(\boldsymbol{\Gamma})^{\rm H}\mathbf{D}(\boldsymbol{\Gamma})
		\end{bmatrix}
		\right).
	\end{aligned}
	\label{eq:DJ_multiband}
\end{equation}
Analogously to the TDoA case, define
\begin{equation}
	\boldsymbol{\Upsilon}
	= [\,\dot{\boldsymbol{\upsilon}}^{\rm T} \;\; \boldsymbol{\Gamma}^{\rm T}\,]^{\rm T}\quad  \text{and}
	\quad
	\dot{\boldsymbol{\upsilon}} = [\,\tau_0 \;\; R \;\; \phi\,]^{\rm T},
	\label{eq:Upsilon}
\end{equation}
with Jacobian
\begin{equation}
	\mathbf{D}(\boldsymbol{\Upsilon})
	= [\,\mathbf{D}(\dot{\boldsymbol{\upsilon}}) \;\; \mathbf{D}(\boldsymbol{\Gamma})\,],
	\label{eq:D_Upsilon}
\end{equation}
where  
$
\mathbf{D}(\dot{\boldsymbol{\upsilon}})
=
\big[
\mathbf{D}_{0}(\dot{\boldsymbol{\upsilon}})^{\rm T} 
\cdots 
\mathbf{D}_{Q-1}(\dot{\boldsymbol{\upsilon}})^{\rm T}
\big],
$
and $
\mathbf{D}_q(\dot{\boldsymbol{\upsilon}}) 
= 
[\,
\mathbf{S}_q^{(\tau_0)} \;\;
\mathbf{S}_q^{(\phi)} \;\;
\mathbf{S}_q^{(R)}
\,]$, with $\mathbf{S}_q^{(R)}$ and $\mathbf{S}_q^{(\phi)}$ being the same as the ones used in \eqref{eq:D_upsilon2}, for the $q$th band.

The multi-band matrix-form FIM is therefore
\begin{equation}
	\mathbf{J}(\boldsymbol{\Upsilon})
	= \frac{2}{\sigma^{2}}
	\Re\!\left(
	\begin{bmatrix}
		\mathbf{D}(\dot{\boldsymbol{\upsilon}})^{\rm H}\mathbf{D}(\dot{\boldsymbol{\upsilon}}) 
		& \mathbf{D}(\dot{\boldsymbol{\upsilon}})^{\rm H}\mathbf{D}(\boldsymbol{\Gamma}) \\
		\mathbf{D}(\boldsymbol{\Gamma})^{\rm H}\mathbf{D}(\dot{\boldsymbol{\upsilon}})
		& \mathbf{D}(\boldsymbol{\Gamma})^{\rm H}\mathbf{D}(\boldsymbol{\Gamma})
	\end{bmatrix}
	\right),
\end{equation}
which yields the generalization of  \eqref{eq:crb_R_phi_no_cf} to multi-band as
\begin{equation}
	{\rm CRB}_{R} = [\,\mathbf{J}(\boldsymbol{\Upsilon})^{-1}\,]_{2,2},
	\qquad
	{\rm CRB}_{\phi} = [\,\mathbf{J}(\boldsymbol{\Upsilon})^{-1}\,]_{3,3}.
	\label{eq:crb_R_phi_multi_band_no_cf}
\end{equation}

\subsection{Multi-band closed-form CRBs}\label{subsec:multi_band_closed_form_CRBs}
This subsection generalizes the results of Section~\ref{sec:crb_R_phi_approx} (single-band) to the multi-band model.

\begin{proposition}
	Propositions~2, 3, and 4, corresponding to the TDoA CRB, the intermediate AoA-Distance CRB, and the ULA AoA-Distance CRB for the single-band system, extend directly to the multi-band model by replacing $\alpha/\sigma^2$ in the respective expressions with
	\begin{equation}
		{\alpha}_{\rm mb} 
        = \sum_{q} \frac{|\gamma_q|^2}{\sigma_q^2} (2\pi)^2 N_q 
        \left( \frac{(N_q^2 - 1) f_{0_q}^2}{12} + f_{{\rm c}_q}^2 \right).
		\label{eq:alpha_mb}
	\end{equation}
	%
	%which is a weighted linear combination of the per-band $\alpha/\sigma^2$ contributions.
\end{proposition}

\begin{proof}
	See Appendix~\ref{app:proposition_multiband}.
\end{proof}

\subsection{Multi-band 4 Stage Estimator and CRB Achievability}\label{subsec:multi_band_estimator}
\begin{algorithm}[t!]
	\small
	\caption{Multi-Band 4-Stage LM Algorithm}
	\label{alg:multi_band}
	\begin{algorithmic}[1]  % The number tells where line numbering starts
		\Input $\mathbf{y} = [\mathbf{y}_{0,0}^{\rm T}\,\, \mathbf{y}_{1,0}^{\rm T} \cdots \mathbf{y}_{M-1,0}^{\rm T}\,\, \mathbf{y}_{0,1}^{\rm T} \cdots \mathbf{y}_{M-1,Q-1}^{\rm T}]^{\rm T}$	
%		\State $q_{\rm ini} \gets 0$ 
		\State $q_{\rm ini} = \max_{q}(\rho_q/\rho_{\delta_q}^\star)$	
		\Statex $/\star$ \texttt{\textcolor{brown}{Stage-1: Delay for $\mathbf{y}_{:,q_{\rm ini}}$, \eqref{eq:s_tau}}} $\hfill \star/$		
		\For{$m = 0$ to $M-1$}
		\State Initialize $\hat{\boldsymbol{\theta}}_{\tau_m,q_{\rm ini}}^{\rm {\rm ini}} = [\hat{\tau}_{m,q_{\rm ini}}^{\rm {\rm ini}} \,\, \hat{\gamma}_{m,q_{\rm ini}}^{\rm r, ini} \,\, \hat{\gamma}_{m,q_{\rm ini}}^{\rm i, ini}]^{\rm T}$ via FFT-based fix grid search \cite{RimaxThesis} taking $\mathbf{y}_{m,q_{\rm ini}}$ as input
		\State Refine $\hat{\boldsymbol{\theta}}_{\tau_m,q_{\rm ini}}^{\rm LM} = {\rm LM}(\mathbf{y}_{m,q_{\rm ini}},\hat{\boldsymbol{\theta}}_{\tau_m,q_{\rm ini}}^{\rm ini};\mathbf{s}_{{\tau}}(\cdot),\mathbf{D}_{{\tau}}(\cdot))$
		\EndFor	
		\Statex $/\star$ \texttt{\textcolor{brown}{Stage-2: TDoA for $\mathbf{y}_{:,q_{\rm ini}}$, \eqref{eq:s_theta} and \eqref{eq:D_theta}}} $\hfill \star/$
		\State Initialize ${\boldsymbol{\theta}}^{\rm ini}_{q_{\rm ini}}$ using ${\delta}_{m,q_{\rm ini}}^{\rm ini} = \hat{\tau}_{m,q_{\rm ini}}^{\rm LM} - \hat{\tau}_{0,q_{\rm ini}}^{\rm LM}$ 
		\State Refine $\hat{\boldsymbol{\theta}}^{\rm LM}_{q_{\rm ini}} = {\rm LM}(\mathbf{y}_{:,q_{\rm ini}},\hat{\boldsymbol{\theta}}^{\rm ini}_{q_{\rm ini}};\mathbf{s}_{\theta}(\cdot),\mathbf{D}_{\theta}(\cdot))$
		\State Estimate $\hat{\boldsymbol{\theta}}^{\rm LM}_{q_{\rm ini}}$ using the received signal of the $q_{\rm ini}$-th sub-band, $\mathbf{y}_{:,q_{\rm ini}}$, in Stages 1 and 2 of Algorithm \ref{alg:single_band} (Line 7).
		\Statex $/\star$ {\texttt{\textcolor{brown}{Stage-3: Multi-B. TDoA, \eqref{eq:s_Theta} and \eqref{eq:DJ_multiband}}}} $\hfill \star/$
		\State Initialize $\hat{\dot{\boldsymbol{\theta}}}^{\rm ini}$ with $\hat{\tau}_{0,q_{\rm ini}}^{\rm LM}$ and $\hat{\delta}_{m,q_{\rm ini}}^{\rm LM}, \forall m$ of $\hat{\boldsymbol{\theta}}^{\rm LM}_{q_{\rm ini}}$
		\State Initialize $\hat{\mathbf{\Gamma}}^{\rm ini}_{q,q} $ with $(\hat{\gamma}_{m,q}^{\rm r, ini} \,\, \hat{\gamma}_{m,q}^{\rm i, ini})$ obtained via correlation using the initial values of $\hat{\dot{\boldsymbol{\theta}}}^{\rm ini}$ for all $q$
%			\For{$q = 0$ to $Q-1$}
%			\State Initialize $\hat{\mathbf{\Gamma}}^{\rm ini}_{q,q} $ with $(\hat{\gamma}_{m,q}^{\rm r, ini} \,\, \hat{\gamma}_{m,q}^{\rm i, ini})$ obtained via correlation using the initial values of $\hat{\dot{\boldsymbol{\theta}}}^{\rm ini}$ 
%			\EndFor
			\State Initialize $\hat{\boldsymbol{\Theta}}^{\rm ini}	= [(\hat{\dot{\boldsymbol{\theta}}}^{\rm ini})^{\rm T} \,\, (\hat{\boldsymbol{\Gamma}}^{\rm ini})^{\rm T}]^{\rm T}$
			\State Refine $\hat{\boldsymbol{\Theta}}^{\rm LM} = {\rm LM}(\mathbf{y},\hat{\boldsymbol{\Theta}}^{\rm ini};\mathbf{s}_{\Theta}(\cdot),\mathbf{D}_{\Theta}(\cdot))$
		\Statex  $/\star$ \texttt{\textcolor{brown}{Stage-4: Dist./Ang., \eqref{eq:s_Theta} and \eqref{eq:D_Upsilon}}} $\hfill \star/$
		\State Initialize $\hat{\boldsymbol{\Upsilon}}^{\rm ini}$ using $\hat{\boldsymbol{\Theta}}^{\rm LM}$ (Eq. \eqref{eq:R_phi_hat_ini})
		\State Refine $\hat{\boldsymbol{\Upsilon}}^{\rm LM} = {\rm LM}(\mathbf{y},\hat{\boldsymbol{\Upsilon}}^{\rm ini};\mathbf{s}_{\Theta(\Upsilon)}(\cdot),\mathbf{D}_{\Theta(\Upsilon)}(\cdot))$
		\Output $\hat{\boldsymbol{\Upsilon}}^{\rm LM}$
	\end{algorithmic}
\end{algorithm}
The 4 stage algorithm for estimating the AoA-Distance parameters of the multi-band model 
$\mathbf{s}(\boldsymbol{\Theta})$ in \eqref{eq:s_Theta} is summarized in Algorithm~\ref{alg:multi_band}.  
The received multi-band signal is
\begin{equation}
	\mathbf{y} = \mathbf{s}(\boldsymbol{\Theta}) + \mathbf{w},
\end{equation}
which can be written explicitly as the concatenation of per-antenna, per-band observations
\[
\mathbf{y}
=
[
\mathbf{y}_{0,0}^{\rm T} \;\;
\mathbf{y}_{1,0}^{\rm T} \;\;
\cdots \;\;
\mathbf{y}_{M-1,0}^{\rm T} \;\;
\mathbf{y}_{0,1}^{\rm T} \;\;
\cdots \;\;
\mathbf{y}_{M-1,Q-1}^{\rm T}
]^{\rm T}.
\]

Following the single-band approach, each stage employs the iterative LM algorithm for refinement.  
Our ultimate goal is to estimate
$
\hat{\boldsymbol{\Upsilon}}^{\rm LM}
=
{\rm LM}(
\mathbf{y}, \hat{\boldsymbol{\Upsilon}}^{\rm ini};
\mathbf{s}_{\Theta(\Upsilon)}(\cdot),\, \mathbf{D}_{\Theta(\Upsilon)}(\cdot)
)
$,
where the initialization $\hat{\boldsymbol{\Upsilon}}^{\rm ini}$ is obtained by applying \eqref{eq:R_phi_hat_ini} to the refined multi-band TDoA estimate
$
\hat{\boldsymbol{\Theta}}^{\rm LM}
=
{\rm LM} (
\mathbf{y},  \hat{\boldsymbol{\Theta}}^{\rm ini};
\mathbf{s}_{\Theta}(\cdot),\, \mathbf{D}_{\Theta}(\cdot)
)
$.
Thus, the key modification in the multi-band setting is the construction of the initial estimate $\hat{\boldsymbol{\Theta}}^{\rm ini}$.

Note that the single-band and multi-band models share the same TDoAs. 
The only difference lies in the channel coefficients, $\gamma$ (single-band) versus $\boldsymbol{\Gamma}$ (multi-band).  
Since each per-band gain $\gamma_q$ can be obtained by correlating the unit-gain TDoA model ($\gamma=1$) with the received signal, a natural strategy is to refine the TDoAs using only the most favorable sub-band, and then infer the remaining channel coefficients via correlation.  

According to Proposition~5, the single-band TDoA estimate in Stage-2 of Algorithm~\ref{alg:single_band} achieves the CRB when $\rho_q > \rho_{\delta_q}^\star$.  
Let
\begin{equation}
q_{\rm ini} = \arg\max_{q} ( \rho_q/\rho_{\delta_q}^\star )
\end{equation}
denote the sub-band with the largest margin above its TSNR threshold.  
The TDoAs are initialized by running Stages~1 and 2 of Algorithm~\ref{alg:single_band} on the $q_{\rm ini}$th subband, yielding $\hat{\boldsymbol{\theta}}^{\rm LM}_{q_{\rm ini}}$.

Stage~3 of Algorithm~\ref{alg:multi_band} then uses $\hat{\boldsymbol{\theta}}^{\rm LM}_{q_{\rm ini}}$ to initialize the multi-band TDoAs and all channel coefficients (via correlation), forming $\hat{\boldsymbol{\Theta}}^{\rm ini}$, and computes the refined estimate
$
\hat{\boldsymbol{\Theta}}^{\rm LM}
=
{\rm LM}(\mathbf{y}, \hat{\boldsymbol{\Theta}}^{\rm ini}; \mathbf{s}_{\Theta}(\cdot), \mathbf{D}_{\Theta}(\cdot)).
$

Finally, Stage~4 refines the AoA–distance parameters, $\hat{\boldsymbol{\upsilon}}^{\rm LM}$, using the same procedure as Stage~3 of Algorithm~\ref{alg:single_band}.

\subsubsection{Multi-band TDoA TSNR}
Since the sub-band $q_{\rm ini}$ is selected to initialize the refined multi-band TDoA, its TSNR condition follows directly from \eqref{eq:rho_delta_star} in Proposition~5, namely 
$\rho_{q_{\rm ini}} > \rho_{\delta_{q_{\rm ini}}}^\star$, 
where $f_{{\rm c}_{q_{\rm ini}}}$, $N_{q_{\rm ini}}$, and 
$B_{q_{\rm ini}} = N_{q_{\rm ini}} f_{0_{q_{\rm ini}}}$ 
are the parameters of the selected sub-band.
Using the same structure as \eqref{eq:rho_delta_star} for, the multi-band TDoA TSNR becomes
\begin{equation}
	\rho_{q_{\rm ini}} 
	> 
	\rho_{\delta_{q_{\rm ini}}}^\star
	\;\approx\;
	\frac{12}{N_{q_{\rm ini}}}\,
	\frac{f_{{\rm c}_{q_{\rm ini}}}^{2}}{B_{q_{\rm ini}}^{2}}.
	\label{eq:rho_delta_star_mb}
\end{equation}

\subsubsection{Multi-band Distance TSNR}
Following the same steps as the single-band derivation in Appendix~\ref{app:dist_waterfall}, and incorporating the multi-band TDoA TSNR, the distance TSNR must satisfy the simultaneous conditions
\begin{equation}
	\left\{
	\begin{array}{l}
		\displaystyle 
		\sum_{q} \rho_q N_q f_{{\rm c}_q}^2 
		> 
		\frac{(4 R c)^2}{(M-1)^4 \delta_{\rm x}^4 \sin^4\phi}, \\[6pt]
		\displaystyle 
		\rho_{q_{\rm ini}} > \rho_{\delta_{q_{\rm ini}}}^\star
	\end{array}
	\right..
	\label{eq:rho_R_multi_band_general}
\end{equation}
Unlike the single-band TSNR expression in \eqref{eq:rho_R_star}, the multi-band case does not allow the left-hand side of the top equation in \eqref{eq:rho_R_multi_band_general} to be expressed solely in terms of a single SNR value.  
Instead, it naturally appears as a linear combination of the per-band SNRs, weighted by $N_q$ and $f_{{\rm c}_q}^2$.  
This highlights that sub-band power allocation plays a key role in meeting the multi-band TSNR, since $\rho_q \propto |\gamma_q|^2$ can be adjusted per sub-band under a total power budget.

\subsubsection{Dual-band Distance TSNR with Dependent Channels}
A particularly instructive case is a two-band system ($q=\{0,1\}$) with  
$f_{{\rm c}_0} < f_{{\rm c}_1}$,  
$|\gamma_0|^2 = |\gamma_1|^2$,  
$N_0 = N_1$,  
and equal sub-carrier spacing.  
In this setting, the TSNR reduces to a single–band–like condition
\begin{equation}
	\rho_0 
	\!> \!
	\rho_{{\rm R}_0}^\star 
	\!=\!
	\frac{1}{N_0}\!
	\max\!\left(\!
	\frac{(4Rc)^2}{(f_{{\rm c}_0}^2 \! + \! f_{{\rm c}_1}^2)(M\!- \! 1)^4 \delta_{\rm x}^4 \sin^4\!\phi},
	\frac{12 f_{{\rm c}_0}^2}{B_0^2}\!
	\right).
	\label{eq:rho_R_star_mb}
\end{equation}
Thus, when the per-band gains $|\gamma_q|^2$ are not freely chosen (e.g., when expressed relative to a reference sub-band $q=0$), the TSNR can be characterized in terms of the reference SNR~$\rho_0$.
%Equation \eqref{eq:rho_R_star_mb} reveals a useful insight into multi-band localization, namely, introducing a higher-frequency sub-band reduces the required TSNR only when the first argument of the $\max$ operator dominates.   
%Lastly, we notice that equation \eqref{eq:rho_R_star_mb} permits a generalization for arbitrary subcarrier spacing, number of sub-carriers, and $|\gamma_0|^2 = K |\gamma_1|^2$ for real valued constant $K$. This general formulation shall be exploited in future work.

% In contrast, when the second argument dominates, including additional higher-frequency bands provides no improvement in this regard.

\begin{figure*}
	\centering \includegraphics[scale=0.43]{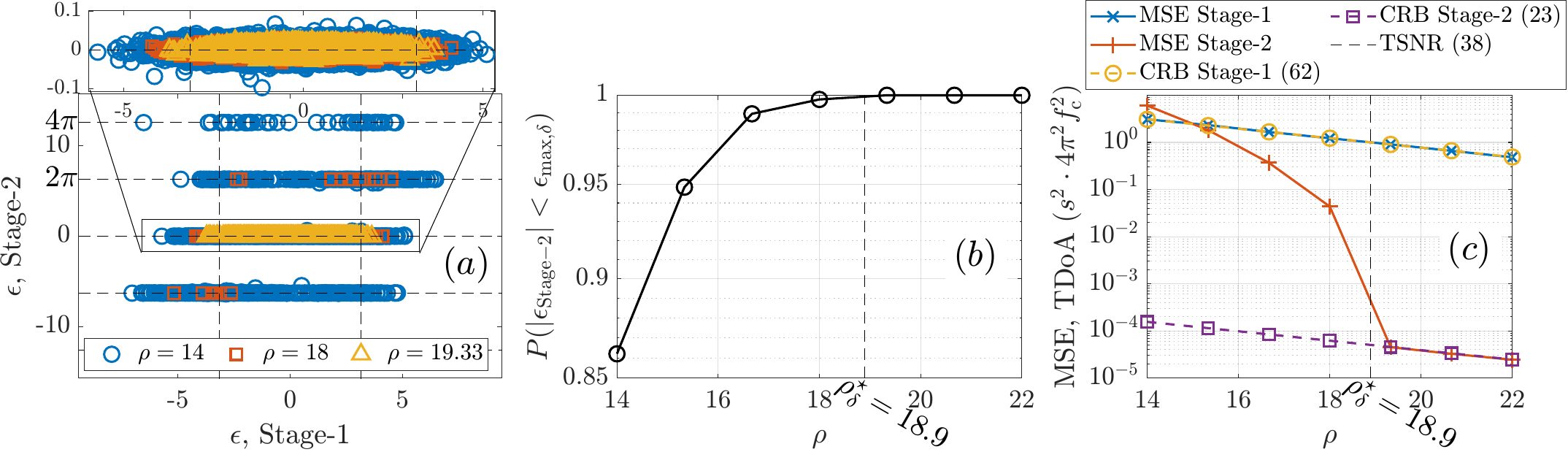} % single column
	\caption{TDoA Results. $(a)$ TDoA error of Stage-1 vs Stage-2. $(b)$ Probability of refined TDoA (Stage-2) being in the correct region. $(c)$ MSE, CRBs, and TSNR. The TDoAs are normalized (divided) by $1/(2 \pi f_{\rm c})$.}
	\label{fig:tdoa_singleband}
       \vspace{-0.4cm}
\end{figure*}

\section{Numerical Results}\label{sec:results}
The goals of the following results are to:  
i) provide numerical evidence and intuition for the TDoA and distance TSNRs in Propositions~5 and 6,  
ii) compare the matrix-form and closed-form CRBs, and  
iii) demonstrate that the single- and multi-band estimators approach the CRB at high SNR.

We first present the single-band results for TDoA, distance, and AoA in Subsections~\ref{subsec:results_single_band_tdoa}, \ref{subsec:results_single_band_distance}, and \ref{subsec:results_single_band_aoa}.  
The multi-band results are then shown in Subsection~\ref{subsec:results_multi_band}.

\begin{table}[t!]
	\centering  
	\caption{Parameters for single-band results.}
	\vspace{-0.2cm}    
	\label{tab:parameters_singleband}	
	\begin{tabular}{lc|lc}    
		\toprule
		  \sizefont Parameter & \sizefont Value & 	\sizefont Parameter & \sizefont Value \\
		\midrule
		\sizefont TX pos., $(x_{\rm T},y_{\rm T})$					& \sizefont $(15,20)\SI{ }{\m}$  &	\sizefont Dist., $R$ & \sizefont $ \SI{25.5}{\m}$\\
		\sizefont Center freq., $f_{\rm c}$ & \sizefont $\SI{10}{\GHz}$ & \sizefont AoA, $\phi$ & \sizefont $ 78.69^\circ$ \\
		\sizefont N. subc., $N$   & \sizefont $256$ &	\sizefont Subc. spacing, $f_0$  & \sizefont $\SI{960}{\kHz}$ \\
		\sizefont N. ant., $M$   & \sizefont $5$ &	\sizefont Ant. spacing, $\delta_{{\rm x}}$   & \sizefont $\SI{0.03}{\m}$ \\		
		\sizefont TDoA TSNR, $\rho_{\delta}^\star$ & \sizefont $\SI{18.9}{\dB}$ &	\sizefont Dist. TSNR, $\rho_{R}^\star$ & \sizefont $\SI{22.8}{\dB}$ \\
		\bottomrule
	\end{tabular}
	%	\vspace{-0.2cm}
\end{table}

\subsection{Single-band TDoA}\label{subsec:results_single_band_tdoa}
The single-band TDoA results are shown in Fig.~\ref{fig:tdoa_singleband} using the parameters in Table~\ref{tab:parameters_singleband}.  
For numerical stability, all TDoAs are normalized by $1/(2\pi f_{\rm c})$, in which case the maximum error limit in \eqref{eq:eps_delta} becomes $\epsilon_{\max,\delta} = \pi$.  
In the plots, all TDoAs for $m = 1,\ldots,M-1$ are displayed as scattered points in graph~$(a)$, while graph~$(b)$ and $(c)$ show averaged results.

The left-most graph~$(a)$ compares the Stage-1 and Stage-2 TDoA errors for identical received signals across multiple realizations.  
Since Stage-2 refines the Stage-1 output, this comparison verifies the refinement on a per-realization basis.  
Three SNR values are shown, $14$, $18$, and $\SI{19.33}{\dB}$, that are all below the TSNR $\rho_{\delta}^\star = \SI{18.9}{\dB}$.  
The Stage-2 results cluster around $\pm k 2\pi$, $k=0,1,\ldots$, reflecting the ambiguity introduced by the constraint 
$b(\hat{\delta}_m) = b(\delta_m)\, b(\pm \epsilon)$, as discussed in Appendix \ref{app:tdoa_waterfall}.  
Because $b(\tau)=\exp(-j 2\pi \tau f_{\rm c})$ in \eqref{eq:ab}, a $2\pi$ phase shift in $b(\pm\epsilon)$ arises when $\epsilon$ changes by integer multiples of $1/f_{\rm c}$, which equals $2\pi$ in normalized units.  
Most importantly, ambiguities occur more often when the Stage-1 TDoA error is close to $\pm \epsilon_{\max,\delta}$.  
At lower SNRs, even small Stage-1 errors may lead to $\pm 2\pi$ jumps in Stage-2 due to larger LM step sizes.  
At $\rho=\SI{20}{\dB}$, which is close to $\rho_{\delta}^\star$, ambiguities occur only when the initialization error is near $\pm \epsilon_{\max,\delta}$. Higher SNR reduces LM step sizes and prevents unnecessary $\pm 2\pi$ shifts when the initial error is small.

Graph~$(b)$ complements graph~$(a)$ by showing the probability that the absolute Stage-2 TDoA error is below $\epsilon_{\max,\delta}=\pi$.  
At $\rho=\SI{20}{\dB}$, this probability is close to one, indicating a low likelihood of $\pm 2\pi$ shifts.  
%For $\rho > \rho_{\delta}^\star$, the probability becomes even closer to one, confirming the absence of ambiguity.  
According to Proposition~5 in Appendix \ref{app:tdoa_waterfall}, this happens because Stage-2 successfully incorporates the $f_{\rm c}$ constraint through $b(\cdot)$ without error. % clarifying why the CRB is achieved for $\rho>\rho_{\delta}^\star$.

Graph~$(c)$ shows the Stage-1 and Stage-2 TDoA MSEs.  
The CRB for Stage-2 is achieved precisely when $\rho>\rho_{\delta}^\star$, numerically validating the rationale of Proposition~5 and the intuition provided by graphs~$(a)$ and~$(b)$.  
Moreover, Stage-2 improves the TDoA MSE by approximately four orders of magnitude, a crucial gain for accurate TDoA-based distance estimation.
Lastly, the CRB expressions in \eqref{eq:C_delta}, \eqref{eq:C_delta_tilde}, as well as the outputs of Stages~1 and 2 of Algorithm~\ref{alg:single_band} are validated.

%In summary, the results depicted in Fig.~\ref{fig:tdoa_singleband} validate multiple theoretical results. The TSNR results of \eqref{eq:rho_delta_star} are analyzed systematically via instructive intermediate results in graphs~$(a)$ and $(b)$ and the RMSE results in graph~$(c)$. The CRB expressions in \eqref{eq:C_delta} (Stage-2) and \eqref{eq:C_delta_tilde} (Stage-1), as well as the efficiency of Algorithm~\ref{alg:single_band}, are verified in graph~$(c)$.

\begin{figure*}
	\centering \includegraphics[scale=0.43]{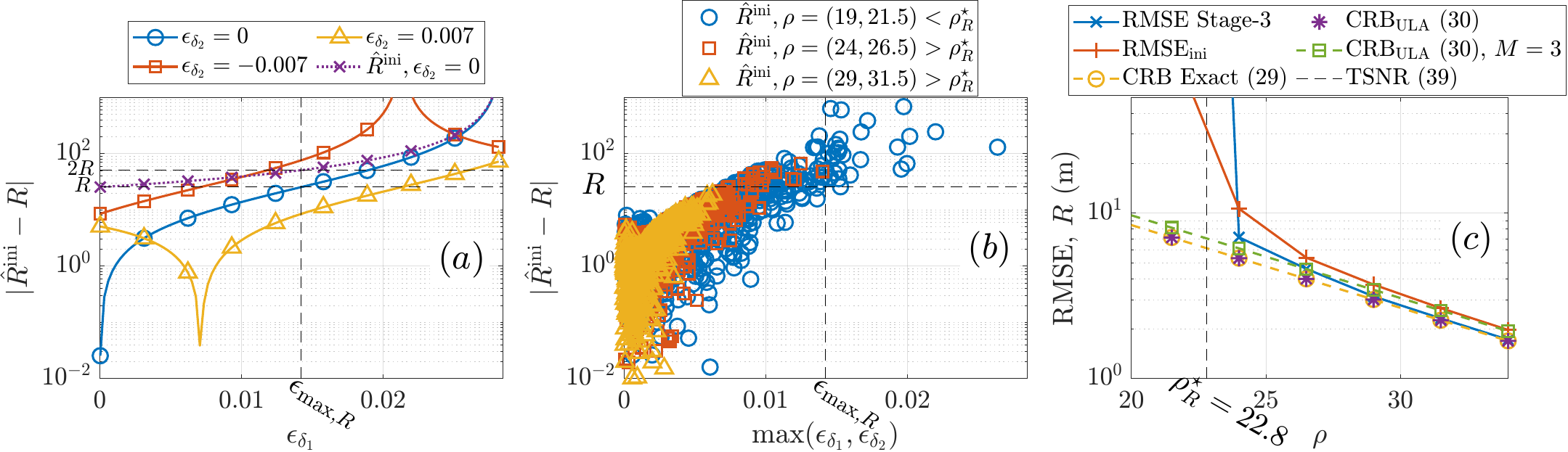} % single column
	\caption{Single-band Distance: $(a)$ Deterministic TDoA error (Stage-2) vs distance initialization error. $(b)$ Simulated TDoA error (Stage-2) vs distance initialization error. $(c)$  RMSE, CRBs, and TSNR.}
	\label{fig:dist_singleband}
       \vspace{-0.5cm}
\end{figure*}

\subsection{Single-band Distance}\label{subsec:results_single_band_distance}
The single-band distance estimation results are shown in Fig.~\ref{fig:dist_singleband} on the next page using the parameters of Table~\ref{tab:parameters_singleband}.

The left-most graph~$(a)$ illustrates the distance initialization error obtained from \eqref{eq:R_phi_hat_ini} under deterministic TDoA errors from Stage-2, where $\epsilon_{\delta_1}$ and $\epsilon_{\delta_2}$ denote the TDoA errors for the two farthest antennas from the reference element at the origin.  
We fix $\epsilon_{\delta_2}=\{0,-0.013,0.013\}$ and vary $\epsilon_{\delta_1}\in[0,0.05]$, also plotting the resulting $\hat{R}_{\rm ini}$ for reference.  
For $\epsilon_{\delta_2}=0$, this setup matches the scenario in \eqref{eq:R_ini} (approximate $\hat{R}_{\rm ini}$ expression):  
$\epsilon_{\delta_1}=\epsilon_{\max,R}$ corresponds to $\eta=-1$, meaning that the TDoA error equals one quadratic-term unit  
$\delta_{\rm quad}=\delta_{\rm x}^2 \sin^2\phi/(2Rc)$.  
As predicted by \eqref{eq:R_ini}, $\epsilon_{\delta_1}=\epsilon_{\max,R}$ yields $\hat{R}_{\rm ini}\approx 2R$, implying an initialization error of order $R$.  
When $\epsilon_{\delta_2}\neq 0$, the resulting errors has a more convoluted form, where $\epsilon_{\delta_2}>0$ reduces the error and $\epsilon_{\delta_2}<0$ increases it.  
%Nevertheless, the $\epsilon_{\delta_2}=0$ case provides a representative characterization of the initialization behavior when $\epsilon_{\delta_1}$ approaches $\epsilon_{\max,R}$.

The center graph~$(b)$ shows scattered data relating the maximum of $\{\epsilon_{\delta_1},\epsilon_{\delta_2}\}$ to the distance initialization error.  
For SNRs below $\rho_R^\star$ (blue circles), many realizations exceed the $(\epsilon_{\max,R},\, R)$ limits, indicating poor initialization since the induced error is on the order of $R$.  
For SNRs above $\rho_R^\star$, the number of such occurrences decreases significantly, though unfavorable combinations of $(\epsilon_{\delta_1},\epsilon_{\delta_2})$ may still produce outliers, consistent with graph~$(a)$.

The right-most graph~$(c)$ displays the RMSE and square-root CRB for the distance estimate.  
For comparison, we also include the $M=3$ CRB, which is the appropriate benchmark for SNR values near $\rho_R^\star$ because the initialization relies on the outer-most three antennas, or equivalently, a sub-array with $3$ antennas.  
Indeed, the RMSE nearly matches the $M=3$ CRB at $\rho=\SI{24}{\dB}>\rho_R^\star$, confirming the rationale of Proposition~6 outlined in Appendix~\ref{app:dist_waterfall}.  
The full-array ($M=5$) CRB is reached only at higher SNRs (above \SI{30}{\dB}), since antennas closer to the array center require a higher SNR to accurately capture the distance-dependent quadratic term.  
%Graph~$(c)$ also highlights the improvement offered by the Stage~3 refinement relative to Stage-2, which is most pronounced just above $\rho_R^\star$.
%(e.g., a $\SI{2}{m}$ RMSE reduction at $\rho=\SI{24.7}{\dB}$ versus $\SI{0.2}{m}$ at $\rho=\SI{34}{\dB}$).

Finally, the results confirm that the approximate ULA CRB in \eqref{eq:C_R_ULA} closely matches the exact matrix-form CRB in \eqref{eq:CRB_phi_R}, validating Proposition~4.

%In summary, the results depicted in Fig.~\ref{fig:dist_singleband} validate multiple theoretical results. The TSNR results of \eqref{eq:rho_R_star} are analyzed systematically via instructive intermediate results in graphs~$(a)$ and $(b)$ and the RMSE results in graph~$(c)$. The closed-form ULA CRB \eqref{eq:CRB_phi_R} as well as the efficiency of Algorithm~\ref{alg:single_band} are verified in graph~$(c)$.

\subsection{Single-band AoA}\label{subsec:results_single_band_aoa}
The AoA results are shown in Fig.~\ref{fig:aoa_singleband} for the parameters in Table~\ref{tab:parameters_singleband}.  
An interesting result is shown: the RMSE of Stage-3 has worse performance than the initialization for SNRs below $\rho_R^*$. Since Stage-3 jointly estimates AoA and distance, the poor distance estimation degrades the AoA estimation, and the CRB is achieved only for SNRs above $\rho_R^*$ with a slight improvement in relation to the initialization. 
This result reveals that, for SNR above $\rho_R^*$, the proposed algorithm does not play a significant role in AoA estimation, because a reasonably good estimate can be achieved assuming the typical far field model. 
%Several results related to the AoA parameter are validated, namely, the AoA estimator achieves the CRB at high SNR, the approximate ULA CRB in \eqref{eq:C_phi_ULA} closely matches the exact matrix-form CRB in \eqref{eq:CRB_phi_R}, and the AoA initialization from \eqref{eq:R_phi_hat_ini} attains an MSE close to the CRB for SNR values above the TDoA TSNR.

\begin{figure}[t!]
	\centering \includegraphics[scale=0.44]{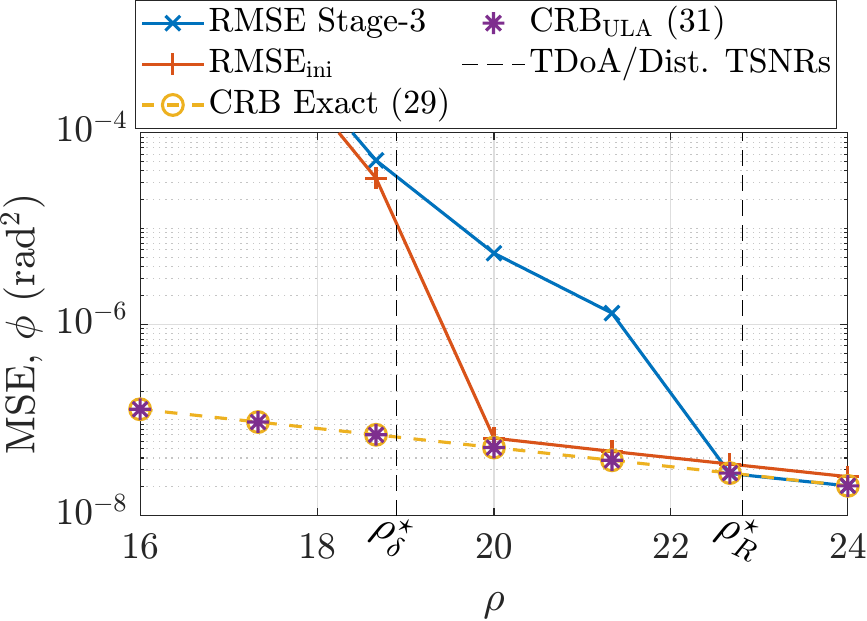} % single column
	\caption{Single-band AoA: MSE, CRBs, and TSNRs.}
	\label{fig:aoa_singleband}
     \vspace{-0.5cm}
\end{figure}

\begin{figure}[t!]
	\centering \includegraphics[scale=0.44]{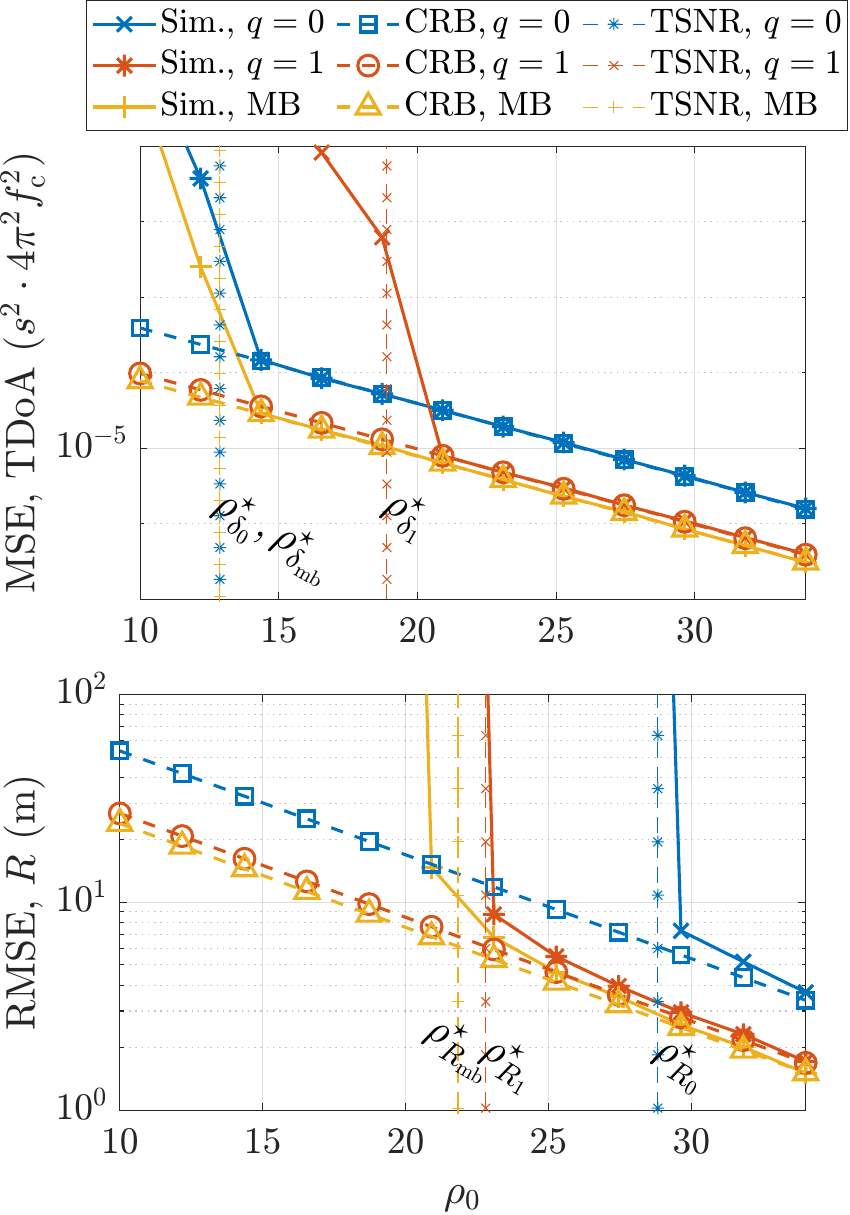} % single column
	\caption{Multi-band TDoA and Distance.}
	\label{fig:multiband}
    \vspace{-0.6cm}
\end{figure}

\subsection{Multi-band TDoA and Distance}\label{subsec:results_multi_band}
The multi-band results are shown in Fig.~\ref{fig:multiband}.  
To isolate the impact of center frequency, we consider two sub-bands, $q=\{0,1\}$, with 
$f_{{\rm c}_0}=\SI{5}{\GHz}$ and $f_{{\rm c}_1}=\SI{10}{\GHz}$, while all remaining parameters match those in Table~\ref{tab:parameters_singleband}.
For the simulation, Algorithm \ref{alg:multi_band} is evaluated. The single- and multi-band TDoA CRBs are computed from \eqref{eq:C_delta} with $\alpha$ and $\alpha_{\rm mb}$, from \eqref{eq:alpha} and \eqref{eq:alpha_mb}, respectively.
The single- and multi-band distance CRBs are computed using \eqref{eq:C_R_ULA} with $\alpha$ and $\alpha_{\rm mb}$. 
The TSNRs for single- and multi-band TDoA are computed from \eqref{eq:rho_delta_star} and \eqref{eq:rho_delta_star_mb}, respectively. 
The single- and multi-band distance TSNRs are computed from \eqref{eq:rho_R_star} and \eqref{eq:rho_R_star_mb}, respectively.

The TDoA results in the top graph clearly illustrate the role of each sub-band in both estimation error and TSNR.  
The lower-frequency sub-band exhibits a larger MSE but a smaller TSNR, whereas the higher-frequency sub-band achieves a lower MSE but requires a significantly higher TSNR.  
%This behavior is consistent with the theoretical expressions developed in this work.  
For example, the single-band TSNR in \eqref{eq:rho_delta_star} scales with $f_{\rm c}^{2}$, implying that band $q=0$ requires \SI{6}{\dB} less SNR than band $q=1$, since $(f_{{\rm c}_0}/f_{{\rm c}_1})^{2}=1/4$.  
Likewise, in \eqref{eq:alpha}, the term $f_{\rm c}^{2}$ dominates $(N^{2}-1)f_{0}^{2}/12\approx B^{2}/12$, so the CRB and the observed MSE scale approximately as $1/f_{\rm c}^{2}$.  
This leads to an instructive duality: TSNR increases with $f_{\rm c}^{2}$, whereas the MSE decreases with $1/f_{\rm c}^{2}$.

Additionally, Algorithm~\ref{alg:multi_band} successfully combines the stronger features of each sub-band: it achieves the lower TSNR requirement of the low-frequency band while attaining an MSE slightly below that of the high-frequency band.  
In this example, the multi-band WSRN is the same as the sub-band for $q_{\rm ini} = \arg\max_{q} ( \rho_q / \rho_{\delta_q}^\star ) = 0$, while the error is close to that of $q=1$.

The distance results are shown in the bottom graph.  
In contrast to the TDoA case, the lower-frequency band exhibits the \emph{higher} TSNR for distance estimation.  
Here, the first term of the $\max(\cdot,\cdot)$ in \eqref{eq:rho_R_star} dominates, and since it scales with $1/f_{\rm c}^{2}$, the TSNR decreases at higher frequencies by \SI{6}{\dB}.  
Importantly, a smaller TDoA TSNR for a low-frequency band does \emph{not} imply a smaller distance TSNR.  
For example, for $q=0$ we have $\rho_{\delta_0}^\star = \SI{12.8}{\dB}$. Because the first term in \eqref{eq:rho_R_star} is much larger, the distance TSNR is $\rho_{R_0}^\star=\SI{28.8}{\dB}$, with a \SI{16}{\dB} gap from $\rho_{\delta_0}^\star$.  
Although the TDoA TSNR increases by \SI{6}{\dB} when moving from $f_{{\rm c}_0}$ to $f_{{\rm c}_1}$, the dominant term in \eqref{eq:rho_R_star} remains unchanged, and the distance TSNR decreases by the same \SI{6}{\dB}, giving 
$\rho_{R_1}^\star=\SI{22.8}{\dB}$.  
%If the TDoA term were dominant, the opposite trend would arise.

Regarding multi-band TSNR, the results confirm that it decreases further as predicted by \eqref{eq:rho_R_star_mb}.
In this case, the TSNR scales with $1/(f_{{\rm c}_0}^{2}+f_{{\rm c}_1}^{2})$, demonstrating that multi-band processing not only reduces the CRB but also relaxes the required SNR.

%In summary, Fig.~\ref{fig:multiband} validates the multi-band theoretical results derived in Section~\ref{sec:system_model_multi_band}.
%We demonstrated how the multi-band TSNR in \eqref{eq:rho_R_star_mb} (or its more general form \eqref{eq:rho_R_multi_band_general}) and the CRB generalization, via \eqref{eq:alpha_mb} in Proposition~6, can be used to analyze the performance of localization system.

%
\section{Conclusion}\label{sec:conclusion}

This paper presented a unified theoretical and algorithmic framework for single- and multi-band localization under a coherent SIMO model. We derived closed-form and matrix-form CRBs for TDoA, AoA, and distance, developed single- and multi-band benchmark estimators, and introduced TSNR conditions that characterize when CRB-level performance becomes achievable. Key results include exact TDoA CRBs, intermediate AoA-distance CRBs for arbitrary arrays, ULA closed-form approximations, and generalization to multi-band. All the theoretical results were numerically validated.

A central contribution of this work is the systematic characterization of the \emph{Threshold SNR} (TSNR).  
The TSNR expressions derived for both TDoA and distance provide practical thresholds that precisely indicate the transition from unreliable, ``off-the-chart'' estimates to CRB-approaching performance.  
%These expressions reveal how signal parameters, array structure, and multi-band aggregation jointly influence estimator reliability.  
In particular, the numerical results illustrated that multi-band operation improves both RSME and TSNR for the distance estimation.
%To the best of our knowledge, no previous work provided a systematic framework to analytically characterize this effect.

Overall, the closed-form TSNR and CRB expressions developed in this work provide a practical tool for system-level design and optimization of next-generation localization, enabling rapid trade-off analysis across waveform, carrier frequency, bandwidth, aperture, and multi-band aggregation without resorting to extensive Monte Carlo studies.

\appendix

\subsection{Inverse FIM of $\tilde{\boldsymbol{\theta}}_\tau$}\label{app:inv_FIM_theta_delay}
Recalling that $\mathbf{D}(\tilde{\boldsymbol{\theta}}_\tau) = [ {\mathbf{S}}^{(\boldsymbol{\tau})} \,\, \mathbf{S}^{(\tilde{\gamma})} ]$, we obtain
\begin{equation}
	\mathbf{J}(\tilde{\boldsymbol{\theta}}_\tau) 
	= \frac{2}{\sigma^2} 
	\Re\!\left( \mathbf{D}(\tilde{\boldsymbol{\theta}}_\tau)^{\rm H} \mathbf{D}(\tilde{\boldsymbol{\theta}}_\tau) \right) 
	= \frac{2}{\sigma^2}
	\begin{bmatrix}
		\mathbf{X} & \mathbf{Y} \\[3pt]
		\mathbf{Y}^{\rm T} & \mathbf{Z}
	\end{bmatrix}.
	\label{eq:J_theta_tilde}
\end{equation}
The matrix blocks are given by
\begin{equation}
	\begin{aligned}
		\mathbf{X} 
		&\stackrel{(a)}{=} 
		\Re\!\left( ({\mathbf{S}}^{(\boldsymbol{\tau})})^{\rm H} {\mathbf{S}}^{(\boldsymbol{\tau})} \right) 
		= \alpha \cdot \mathbf{I}_M, \\[3pt]
		\mathbf{Y} 
		&\stackrel{(b)}{=} 
		\Re\!\left( ({\mathbf{S}}^{(\boldsymbol{\tau})})^{\rm H} {\mathbf{S}}^{(\tilde{\gamma})} \right) 
		= 2\pi N f_{\rm c} \cdot \mathbf{1}_M 
		\left[ \tilde{\gamma}^{\rm i} \,\, -\!\tilde{\gamma}^{\rm r} \right], \\[3pt]
		\mathbf{Z} 
		&\stackrel{(c)}{=} 
		\Re\!\left( ({\mathbf{S}}^{(\tilde{\gamma})})^{\rm H} {\mathbf{S}}^{(\tilde{\gamma})} \right) 
		= \mathbf{I}_2 \cdot M N.
	\end{aligned}
	\label{eq:XYZ}
\end{equation}
The identity structure of $\mathbf{X}$ in $(a)$ indicates the absence of cross-terms between different delays, which occurs because each delay $\tau_m$ is associated only with the $m$th antenna. 
%In this model, this property arises from the block-diagonal structure of ${\mathbf{S}}^{(\boldsymbol{\tau})}$, resulting from the Khatri–Rao product with the diagonal matrix in $\mathbf{B}(\boldsymbol{\delta}_0)$ (see \eqref{eq:s_theta}). 
Recalling the auxiliary variable $\boldsymbol{\Xi}(f_0, f_{\rm c}) = -j 2\pi (\grave{\mathbf{N}} f_0 + \mathbf{I} f_{\rm c})$ defined subsequently to \eqref{eq:S_delta_derivatives}, the terms multiplying $f_0$ and $f_{\rm c}$ in $\alpha$ are computed using the Identities 1 and 2 of Appendix~\ref{app:useful_identities}.
Regarding $\mathbf{Y}$ in $(b)$, there is no contribution from $f_0$ due to Identity 1 of Appendix~\ref{app:useful_identities}, leaving only the term related to $f_{\rm c}$ present in ${\mathbf{S}}^{(\boldsymbol{\tau})}$. 
The vector associated with the channel gain, $\left[ \tilde{\gamma}^{\rm i} \,\, -\!\tilde{\gamma}^{\rm r} \right]$, arises from the multiplication by the unit imaginary number in $\boldsymbol{\Xi}(f_0, f_{\rm c})$ within ${\mathbf{S}}^{(\boldsymbol{\tau})}$, and the real part operation in \eqref{eq:XYZ}.
Lastly, the quantity $\mathbf{Z}$ in $(c)$ can be written as 
$\mathbf{Z} = MN\,\Re\!\left( [\mathbf{1} \, j\mathbf{1}]^{\rm H} [\mathbf{1} \, j\mathbf{1}] \right)$, 
which has an identity form because the off-diagonal terms are purely imaginary.

Next, the inverse \ac{FIM} in \eqref{eq:J_theta_tilde} is computed via the Schur complement \cite{boyd2004convex} as
\begin{equation}
	\begin{aligned}
		\mathbf{J}(\tilde{\boldsymbol{\theta}}_\tau)^{-1} \! & \! = \frac{\sigma^2}{2} \left[ \begin{matrix} \mathbf{X}^{-1}\! + \!\mathbf{X}^{-1}\mathbf{Y}\mathbf{G}^{-1}\mathbf{Y}^{\rm T} \mathbf{X}^{-1} & \!\!\!-\mathbf{X}^{-1}\mathbf{Y}\mathbf{G}^{-1} \\ -\mathbf{G}^{-1}\mathbf{Y}^{\rm T}\mathbf{X}^{-1} & \mathbf{G}^{-1} \end{matrix} \right] 
		\\& = \frac{\sigma^2}{2} \left[ \begin{matrix} \mathbf{I}_M \frac{1}{\alpha} + \frac{\beta}{\alpha^2} \mathbf{1}_M \mathbf{1}_M^{\rm T} & \cdots \\ \cdots & \ddots \end{matrix} \right]			
	\end{aligned}
	\label{eq:J_theta_delay_app}
\end{equation}
where $\mathbf{G} = \mathbf{Z} - \mathbf{Y}^{\rm T}\mathbf{X}^{-1}\mathbf{Y}$ and 
$
\beta\! = \!(2\pi)^2 N^2 f_{\rm c}^2 
\left( 
[\tilde{\gamma}^{\rm i} \! -\!\tilde{\gamma}^{\rm r}] 
\mathbf{G}^{-1} 
[\tilde{\gamma}^{\rm i} \! -\!\tilde{\gamma}^{\rm r}]^{\rm T}
\right).
$
The second line of \eqref{eq:J_theta_delay_app} follows directly from substituting the quantities in \eqref{eq:XYZ}. 
%The key feature of this formulation lies in the top-left block of the inverse matrix, which contains the \ac{CRB} associated with the delay parameters. 
%This block exhibits a convenient diagonal-plus-rank-one structure that enables a closed-form expression for the \ac{CRB} of the TDoAs.

\subsection{Useful Identities}\label{app:useful_identities}
\begin{itemize}
    \item Identity 1: $\sum_{n=0}^{N-1} \!\big(n - \frac{N-1}{2}\big) = 0$.
    \item Identity 2: $\sum_{n=0}^{N-1} \!\big(n - \frac{N-1}{2}\big)^2 = N(N^2-1)/2$.
    \item Identity 3: $\sum_{n=0}^{N-1} \!\big(n\! -\! \frac{N-1}{2}\big)^4 = N (N^2  -  1)(3N^2 \! - \! 7)/240$.
\end{itemize}
    
\subsection{Closed-form CRB for ULA}\label{app:CRB_ULA}
Computing the closed-form \ac{CRB} in \eqref{eq:CRB_phi_R} requires evaluating the inverse of
\begin{equation}
	\mathbf{U}^{\rm T}\mathbf{U} 
	= 
	\begin{bmatrix}
		M & \mathbf{1}^{\rm T}\boldsymbol{\delta}_0^{(\phi)} & \mathbf{1}^{\rm T}\boldsymbol{\delta}_0^{(R)} \\[3pt]
		\mathbf{1}^{\rm T}\boldsymbol{\delta}_0^{(\phi)} & \| \boldsymbol{\delta_0}^{(\phi)} \|^2 & (\boldsymbol{\delta}_0^{(\phi)})^{\rm T}\boldsymbol{\delta}_0^{(R)} \\[3pt]
		\mathbf{1}^{\rm T}\boldsymbol{\delta}_0^{(R)} & (\boldsymbol{\delta}_0^{(\phi)})^{\rm T}\boldsymbol{\delta}_0^{(R)} & \| \boldsymbol{\delta}_0^{(R)} \|^2
	\end{bmatrix}.
	\label{eq:UUT}
\end{equation}
Our approach is to approximate the TDoAs ($\delta_m$ in \eqref{eq:delta_R_pi}) and their derivatives with respect to $\phi$ and $R$ ($\delta_m^{(\phi)}$ and $\delta_m^{(R)}$) with the linear and quadratic (1st and 2nd) terms of Taylor series expansions around $x_m$.
This yields
\begin{equation}
	\begin{aligned}
		\delta_m &\approx -\frac{x_m \cos \phi}{c} + \frac{x_m^2 \sin^2 \phi}{2 R c}, \\[3pt]
		\delta_m^{(\phi)} 
		& \approx \frac{x_m \sin \phi}{c} 
		\!\left( 1 \!\!+\!\! \frac{x_m \cos \phi}{R} \right) 
		\approx \frac{x_m \sin \phi}{c}, \\[3pt]
		\text{and } \delta_m^{(R)} 
		&	\approx -\frac{x_m^2 \sin^2 \phi}{2 R^2 c}.
	\end{aligned}
	\label{eq:delta_taylor}
\end{equation}
In the expression for $\delta_m^{(\phi)}$, we used $\partial(\sin^2 \phi)/\partial \phi = 2 \sin \phi \cos \phi$ (chain rule) and assume $x_m / R \ll 1$ for the array centered around $x=0$. 

Using Identity 1 of Appendix~\ref{app:useful_identities}, we found that the cross terms associated with the angle $\phi$ approximately zero, i.e., $\mathbf{1}^{\rm T}\boldsymbol{\delta}^{(\phi)} \approx 0$ and 
$(\boldsymbol{\delta}^{(\phi)})^{\rm T}\boldsymbol{\delta}^{(R)} \approx 0$.
Hence, the \ac{CRB} for the AoA simplifies to $C(\phi) = \frac{\sigma^2}{2\alpha} [\mathbf{U}^{-1}]_{2,2} = 1/\|\boldsymbol{\delta}^{(\phi)}\|^2$ yielding \eqref{eq:C_phi_ULA} using Identity 2 of Appendix~\ref{app:useful_identities}.

For the distance, equation \eqref{eq:C_R_ULA} is found by solving $C_{\rm ULA}(R) = \frac{\sigma^2}{2\alpha} 
		\frac{M}{{\rm det}(\mathbf{U}_{\tau_0,R}^{\rm T}\mathbf{U}_{\tau_0,R})}$, where
        \[
\mathbf{U}_{\tau_0,R}^{\rm T}\mathbf{U}_{\tau_0,R} 
= 
\begin{bmatrix}
	M & \mathbf{1}^{\rm T}\boldsymbol{\delta}^{(R)} \\[3pt]
	\mathbf{1}^{\rm T}\boldsymbol{\delta}^{(R)} & \| \boldsymbol{\delta}^{(R)} \|^2
\end{bmatrix}
\]
is a sub-matrix of $\mathbf{U}^{\rm T}\mathbf{U}$ excluding the terms associated with $\phi$, whose elements are computed using the Identities 2 and 3 of Appendix~\ref{app:useful_identities}.

\subsection{TDoA Threshold SNR $\rho_{\delta}^\star$}\label{app:tdoa_waterfall}
Consider the initial TDoA estimate $\hat{\delta}_m = \hat{\tau}_m - \hat{\tau}_0$, obtained by the delay estimates of Stage-1, that is used for initialization of Stage-2.
Recall that the model used in Stage-1 is described in \eqref{eq:s_tau} and ignores the dependence of $b(\hat{\delta}_m)$ on $\hat{\delta}_m$.  
On the contrary, the model of Stage-2 is described in \eqref{eq:s_theta}, which has the constraint $b(\cdot)$ (in $\mathbf{B}(\cdot)$) as a fundamental distinguishing term from \eqref{eq:s_tau}.
Then, we need to analyze how sensitive the estimation error from Stage-1 is in the constraint $b(\cdot)$.
Let the estimation error be $\hat{\delta}_m = \delta_m + \epsilon$, yielding
\begin{equation}
b(\hat{\delta}_m) = \underbrace{\exp(-j 2 \pi \delta_m f_{\rm c})}_{\text{correct phase}}\, \underbrace{\exp(-j 2 \pi \epsilon f_{\rm c})}_{\text{error}}.
\label{eq:b_hat}
\end{equation}
To impose a maximum unambiguous phase error, we require  
$b(\pm \epsilon_{\max,\delta}) = \exp(\mp j\pi)$, corresponding to a maximum phase error of $\mp \pi$.
This leads to the limit
\begin{equation}
	\epsilon_{\max,\delta} = 1/(2 f_{\rm c}).
	\label{eq:eps_delta}
\end{equation}
Our hypothesis is that if $|\epsilon| < \epsilon_{\max,\delta}$ with high probability, then the TDoA estimates from Stage-1 are sufficiently adjusted to the Stage-2 model due to confined phase error within unambiguous limits.

Assuming $\epsilon$ is approximately Gaussian, most of its probability mass lies within a few standard deviations of the mean.
We use this well-known fact to calculate the standard deviation of $\epsilon$ such that $|\epsilon| < \epsilon_{\max,\delta}$ with high probability.
It turns out that a factor of $\pi$ standard deviations is a convenient numerical choice\footnote{A typical rule of thumb considers 3 standard deviations to bound the most likely region of Gaussian events. We used the same idea, where we replaced $3$ by $\pi$ to simplify the equations.}, which translates into defining $\epsilon_{\max,\delta}/\pi$ as the maximum standard deviation of $\epsilon$ to bound the TDoA error with high probability.
In other words, if $\epsilon^\star \sim \mathcal{N}(0,(\epsilon_{\max,\delta}/\pi)^2)$ denotes the TDoA error at the TSNR, then  
$P(|\epsilon^\star| < \epsilon_{\max,\delta}) = 0.9983$. 
%Thus, for SNRs above the TSNR threshold, we expect  
%$P(|\epsilon| < \epsilon_{\max,\delta}) > 0.9983$, making $(\epsilon_{\max,\delta}/\pi)^2$ a logical variance reference.
%In summary, our hypothesis states that when the TDoA MSE (using the $f_{\rm c}=0$ model) satisfies $E(\epsilon^2) < \epsilon_{\max,\delta}^2/(\pi^2) = 1/(2\pi f_{\rm c})^2$, the TDoA estimate from Stage-1 is accurate enough for Stage-2 to refine it further. 

Then, using the SNR definition $\rho = |\gamma|^2/\sigma^2$, and TDoA CRB from Stage-1 parametrized by $f_{\rm c}=0$
\begin{equation}
	\tilde{C}(\delta_m)
	= C(\delta_m \mid f_{\rm c}=0)
	= \frac{\sigma^2}{|\gamma|^2 (2\pi)^2 N}
	\frac{12}{f_0^2 (N^2 - 1)},
	\label{eq:C_delta_tilde}
\end{equation}
we solve the inequality $\tilde{C}(\delta_m) < 1/(2\pi f_{\rm c})^2$ for $\rho$ yielding the TSNR threshold in \eqref{eq:rho_delta_star}.

\subsection{Distance Threshold SNR $\rho_{R}^\star$}\label{app:dist_waterfall}
The Taylor approximation of the TDoA in \eqref{eq:delta_taylor} contains a linear and a quadratic term in the antenna position $x_m$.  
Noticing that the quadratic term, $\delta_{m,{\rm quad}} = \frac{x_m^2 \sin^2\phi}{2Rc}$, is the near-field correction that captures the dependence on $R$, we propose the following hypothesis: if the absolute TDoA error $|\epsilon| = |\hat{\delta}_m-\delta_m| < |\delta_{m,{\rm quad}}|$ with high probability, the TDoA estimate retains enough information about $R$ for the LM refinement to improve the distance estimation.

In the following, we provide an instructive example to illustrate the interplay between TDoA and distance errors, and provide a quantitative result that justifies the above hypothesis.
Consider a system with $M=3$ antennas indexed by $m=\{0,1,2\}$ located at $x_0=0$, $x_1=\delta_{\rm x}$, and $x_2=-\delta_{\rm x}$.  
Let the estimated TDoAs be $\hat{\delta}_1 = \delta_1 + \eta |\delta_{m,{\rm quad}}|$ and $\hat{\delta}_2 = \delta_2$,
%\[
%\hat{\delta}_1 = \delta_1 + \eta |\delta_{m,{\rm quad}}|, 
%\qquad 
%\hat{\delta}_2 = \delta_2,
%\]
so that only the first TDoA contains an error, parameterized by $\eta$.  
Using the Taylor approximation \eqref{eq:delta_taylor} with the above TDoAs estimates, and plugging the result into the conversion formulas from TDoA to distance-angle in \eqref{eq:R_phi_hat_ini} yields the distance expression
\begin{equation}
	\hat{R}_{\rm ini}
	= \frac{2R}{2+\eta}
	- \frac{\delta_{\rm quad}(1+\eta+\eta^2/2)}{2+\eta}
	- \frac{\delta_{\rm lin}\eta}{2+\eta}
	\approx \frac{2R}{2+\eta},
	\label{eq:R_ini}
\end{equation}
where $\delta_{\rm quad} = \frac{\delta_{\rm x}^2 \sin^2\phi}{2Rc}$ and $\delta_{\rm lin} = \delta_{\rm x}\cos\phi$.
The approximation follows because term associated with $R$ dominates for $|\eta| \le 1$, as $\delta_{\rm quad}\ll\delta_{\rm lin}\ll R$.  
Setting $\eta=\pm 1$ (i.e., TDoA error equal to $\pm \delta_{\rm quad}$) gives  
$\hat{R}_{\rm ini} \approx \frac{2}{3}R$ and $\hat{R}_{\rm ini}\approx 2R$, showing that a negative quadratic-sized error produces an $\hat{R}_{\rm ini}$ on the order of $R$ itself, leading to poor initialization.  
In summary, the above analysis provides numerical evidence that $|\epsilon| > |\delta_{m,{\rm quad}}|$ leads to a poor initialization for the distance.
Hence, the quantity $|\delta_{m,{\rm quad}}|$ is a reasonable boundary choice for $|\epsilon|$, which is then used to derive the distance TSNR.

The antennas farthest from the center are located at  
$x_{\max} = \pm (M-1)\delta_{\rm x}/2$,  
which substituted into $\delta_{\rm quad}$ yields the maximum tolerable TDoA error
\begin{equation}
	\epsilon_{\max,R}
	= \frac{(M-1)^2 \delta_{\rm x}^2 \sin^2\phi}{8Rc}.
	\label{eq:app_delay_max}
\end{equation}
Next, we seek the SNR for which  
$|\epsilon_m| < \epsilon_{\max,R}$ with high probability, for the farthest antennas from the center, $m=m_1$ and $m_2$.  
Using the CRB $C(\delta)$ in \eqref{eq:C_delta} for the TDoA estimation at Stage-2, the TSNR threshold is obtained by enforcing
\begin{equation}
C(\delta) < \frac{\epsilon_{\max,R}^2}{\pi^2}.
\end{equation}

Adopting the same Gaussian-based rationale used in Appendix~\ref{app:tdoa_waterfall}.  
Substituting $C(\delta)$ from \eqref{eq:C_delta}, $\epsilon_{\max,R}$ from \eqref{eq:app_delay_max}, and $\rho = |\gamma|^2/\sigma^2$ gives
$\rho >
\frac{K-1}{NK}
\left(
\frac{4Rc}{f_{\rm c} (M-1)^2 \delta_{\rm x}^2 \sin^2\phi}
\right)^2
$,
where  $K = 1 + 12 f_{\rm c}^2/((N^2 - 1)f_0^2)$.
Since $K\gg 1$, we have $(K-1)/(NK) \approx 1/N$.

Finally, because Stage-2 achieves its CRB only when $\rho > \rho_{\delta}^\star$ in \eqref{eq:rho_delta_star}, the TSNR for distance estimation must exceed both thresholds.  
Thus, the overall TSNR is given by the maximum of $\rho_{\rm R}'/N$ and $\rho_{\delta}^\star$, leading to the expression in \eqref{eq:rho_R_star}.
\subsection{Multi-band generalization}\label{app:proposition_multiband}
To avoid redundancy, we directly write the modified multi-band model in Subsection~\ref{subsec:modified_model} analogously to the single-band derivation of Subsection~\ref{subsec:modified_model}, where the transformation from the original to the modified single-band delay models,
\eqref{eq:s_m} (parameterized by ${\boldsymbol{\theta}}_\tau$) and
\eqref{eq:s_m_tilde} (parameterized by $\tilde{\boldsymbol{\theta}}_\tau$), is analogously applied to the multi-band model, yielding ${\boldsymbol{\Theta}}_\tau$ and $\tilde{\boldsymbol{\Theta}}_\tau$, respectively.

By combining the multi-band Jacobians in \eqref{eq:D_multiband} and the corresponding FIM
structure in \eqref{eq:DJ_multiband} with the modified-model construction of
Subsection~\ref{subsec:modified_model}, one can show that the FIM of the modified multi-band
delay model, $\tilde{\boldsymbol{\Theta}}_\tau$, has the same block structure as the
single-band counterpart in \eqref{eq:J_theta_tilde}. Formally,
\begin{equation}
	\mathbf{J}(\tilde{\boldsymbol{\Theta}}_\tau) 
	= \frac{2}{\sigma^2} 
	\Re\!\left( \mathbf{D}(\tilde{\boldsymbol{\Theta}}_\tau)^{\rm H} \mathbf{D}(\tilde{\boldsymbol{\Theta}}_\tau) \right) 
	= \frac{2}{\sigma^2}
	\begin{bmatrix}
			\mathbf{X} & \mathbf{Y} \\[3pt]
			\mathbf{Y}^{\rm T} & \mathbf{Z}
		\end{bmatrix},
	\label{eq:J_theta_tilde2}
\end{equation}
where $\mathbf{D}(\tilde{\boldsymbol{\Theta}}_\tau)
= [\mathbf{D}(\dot{\boldsymbol{\theta}}_\tau)\;\;\;
\mathbf{D}(\tilde{\boldsymbol{\Gamma}})]$
denotes the Jacobian of the modified multi-band delay model.
Following the same steps as in Appendix~\ref{app:inv_FIM_theta_delay}, the sub-matrices are
\begin{equation}
	\begin{aligned}
			\mathbf{X} 
			&\stackrel{(a)}{=} \Re\!\left(\sum_q  ({\mathbf{S}}^{(\boldsymbol{\tau})}_q)^{\rm H} {\mathbf{S}}^{(\boldsymbol{\tau})}_q \right) 
			=  \alpha_{\rm mb} \cdot \mathbf{I}_M, \\[3pt]
			\mathbf{Y} 
			&\stackrel{(b)}{=}  \Re\!\left(  \left[({\mathbf{S}}^{(\boldsymbol{\tau})}_0)^{\rm H} {\mathbf{S}}^{(\tilde{\gamma}_0)}_0  \cdots ({\mathbf{S}}^{(\boldsymbol{\tau})}_{Q-1})^{\rm H} {\mathbf{S}}^{(\tilde{\gamma}_{Q-1})}_{Q-1}\right]\right) 
			= 	2\pi \cdot \mathbf{1}_M\\ & 
			\! \! \!\! \begin{bmatrix}
					N_0 f_{{\rm c}_0}  \left[ \tilde{\gamma}^{\rm i}_0 \,\, -\!\tilde{\gamma}^{\rm r}_0 \right] & \!\!\!\! \cdots \!\!\!\! & N_{Q-1} f_{{\rm c}_{Q-1}}  \left[ \tilde{\gamma}^{\rm i}_{Q-1} \,\, -\!\tilde{\gamma}^{\rm r}_{Q-1} \right]		
				\end{bmatrix}, \\[3pt]
			\mathbf{Z} 
			&\stackrel{(c)}{=} 	
            %\Re\!\left( \begin{bmatrix}
			%		({\mathbf{S}}^{(\tilde{\gamma}_0)}_0)^{\rm H} {\mathbf{S}}^{(\tilde{\gamma}_0)}_0 & &
			%		\\ & \ddots & \\
			%		& & ({\mathbf{S}}^{(\tilde{\gamma}_{Q-1})}_{Q-1})^{\rm H} {\mathbf{S}}^{(\tilde{\gamma}_{Q-1})}_{Q-1}
			%	\end{bmatrix} \right) 
			%\\ & = 
            \begin{bmatrix}
					\mathbf{I}_{2} \cdot M N_0 & &
					\\ & \ddots & \\
					& & \mathbf{I}_{2} \cdot M N_{Q-1}
				\end{bmatrix}.
		\end{aligned}
	\label{eq:XYZ2}
\end{equation}
Line $(a)$ follows from the structure of $\mathbf{D}(\dot{\boldsymbol{\theta}}_\tau)$ in
\eqref{eq:D_multiband}, where $\alpha_{\rm mb}$ is the weighted sum of the single-band
contributions. Lines $(b)$ and $(c)$ consider the block-diagonal structure of $\mathbf{D}(\tilde{\boldsymbol{\Gamma}})$.

Crucially, the single-band and multi-band FIMs in \eqref{eq:XYZ} and
\eqref{eq:XYZ2} share two structural properties, namely, $\mathbf{X}$ is proportional to the identity matrix, and $\mathbf{Y}$ contains the all-ones vector $\mathbf{1}_M$.  
These properties ensure that the top-left block of the inverse FIM retains the same form as
\eqref{eq:J_theta_delay_app} given by
\begin{equation}	
	\mathbf{J}(\tilde{\boldsymbol{\Theta}}_\tau)^{-1} = \frac{\sigma^2}{2} \left[ \begin{matrix} \mathbf{I}_M \frac{1}{\alpha_{\rm mb}} + \frac{\beta_{\rm mb}}{\alpha_{\rm mb}^2} \mathbf{1}_M \mathbf{1}_M^{\rm T} & \cdots \\ \cdots & \ddots \end{matrix} \right],
	\label{eq:J_delay_inv_multi_band}
\end{equation}
where $\beta_{\rm mb}$ is the multi-band analogue of $\beta$, obtained from the
factorization $\mathbf{Y}\mathbf{G}^{-1}\mathbf{Y}^{\rm T}
= \mathbf{1}_M\,\beta_{\rm mb}\,\mathbf{1}_M^{\rm T}$.% and whose exact definition is irrelevant for this proof.

Propositions~2, 3, and 4 depend only on linear transformations of the top-left block of \eqref{eq:J_delay_inv_multi_band}. 
Since it has the same structure as the single-band block in \eqref{eq:J_delay}, all cancellations proceed identically, and the multi-band generalization is obtained by replacing
$\alpha$ with $\alpha_{\rm mb}$.

We remark that the above formulation assumes equal noise power per sub-band, which does not hold in general.
The generalization is done by removing the noise power term in \eqref{eq:J_delay_inv_multi_band}, and writing the SNR per subband $|\gamma_q|^2/\sigma_q^2$ in $\tilde{\alpha}_{\rm mb}$ (see \eqref{eq:alpha_mb}). 
In this case, the multi-band CRB expressions replaces $\alpha/\sigma^2$ by $\tilde{\alpha}_{\rm mb}$.

\bibliography{references_ha}{}
\bibliographystyle{IEEEtran}

\vfill

\end{document}